%% file: final.tex
\RequirePackage{amsmath}
\pdfoutput=1
\documentclass{llncs}
\input{glyphtounicode}
\usepackage{lmodern}
\usepackage[T1]{fontenc}
\usepackage[utf8]{inputenc}
\usepackage{amsfonts}
\usepackage{url}
\usepackage[english]{babel}
\usepackage{algorithm}
\usepackage{fullpage}
\usepackage{algpseudocode}
\usepackage{graphicx}
\usepackage{amssymb}

\usepackage{amsthm}

\usepackage{tikz}

\newcommand{\bigo}{\mathcal{O}}

\newcommand{\R}{\mathbb{R}}
\newcommand{\Z}{\mathbb{Z}}

\newcommand{\Zq}{\mathbb{Z}/q\mathbb{Z}}
\newcommand{\UniqueSVP}{\mathsf{UniqueSVP}}
\newcommand{\GapSVP}{\mathsf{GapSVP}}
\newcommand{\BDD}{\mathsf{BDD}}
\mathchardef\mhyphen="2D 
\newcommand{\DecisionLWE}{\mathsf{Decision\mhyphen LWE}}
\newcommand{\SearchLWE}{\mathsf{Search\mhyphen LWE}}
\newcommand{\zeroun}{\left\{0,1\right\}}
\newcommand{\me}{\mathrm{e}}
\newcommand{\E}{\mathbb{E}}
\newcommand{\Var}{\mathrm{Var}}
\newcommand{\vol}{\mathrm{vol}}
\newcommand{\Li}{\mathcal{L}}
\newcommand{\Mod}{\text{ mod }}
\newcommand{\Dis}{\mathcal{D}}
\DeclareMathOperator*{\argmax}{arg\,max}

\usepackage[pdftex,bookmarks,bookmarksopen,bookmarksdepth=3]{hyperref}
\hypersetup{colorlinks=true,citecolor=blue,linkcolor=blue,urlcolor=black}

\title{An Improved BKW Algorithm for LWE \\ 
{with Applications to Cryptography and Lattices}}
\titlerunning{LWE Algorithm and Applications}
\author{Paul Kirchner\inst{1} and Pierre-Alain Fouque\inst{2}}
\institute{ 
   \'Ecole normale sup\'erieure\\
 \and
   Universit\'e de Rennes 1 and
   Institut universitaire de France\\
   \email{$\{$paul.kirchner,pierre-alain.fouque$\}$@ens.fr}}

\begin{document}
\maketitle

\begin{abstract}
In this paper, we study the Learning With Errors problem and its binary variant, where secrets and errors are binary or taken in a small interval. We introduce a new variant of the Blum, Kalai and Wasserman algorithm, relying on a quantization step that generalizes and fine-tunes modulus switching. In general this new technique yields a significant gain in the constant in front of the exponent in the overall complexity. We illustrate this by solving within half a day a $\mathsf{LWE}$ instance with dimension $n=128$, modulus $q=n^{2}$, Gaussian noise $\alpha=1/(\sqrt{n/\pi} \log^2 n)$ and binary secret, using $2^{28}$ samples, while the previous best result based on BKW claims a time complexity of $2^{74}$ with $2^{60}$ samples for the same parameters.\\
We then introduce variants of $\mathsf{BDD}$, $\mathsf{GapSVP}$ and $\mathsf{UniqueSVP}$, where the target point is required to lie in the fundamental parallelepiped, and show how the previous algorithm is able to solve these variants in subexponential time.
Moreover, we also show how the previous algorithm can be used to solve the $\mathsf{BinaryLWE}$ problem with $n$ samples in subexponential time $2^{(\ln 2/2+o(1))n/\log \log n}$. This analysis does not require any heuristic assumption, contrary to other algebraic approaches; instead, it uses a variant 
of an idea by Lyubashevsky to generate many samples from a small number of samples. 
This makes it possible to asymptotically and heuristically break the $\mathsf{NTRU}$ cryptosystem in subexponential time (without contradicting its security assumption). We 
are also able to solve subset sum problems in subexponential time for density $o(1)$, which is of independent 
interest: for such density, the previous best algorithm requires exponential time. As a direct application, we can solve in subexponential time the parameters of a cryptosystem based on this problem proposed at TCC 2010.
\end{abstract}

\section{Introduction}

\input{intro}

\section{Preliminaries}\label{sec:preliminaries}

We identify any element of $\Zq$ to the smallest of its equivalence class, the \
positive one in case of tie.
Any vector $\vec{\mathrm{x}} \in \big(\Zq\big)^n$ has an Euclidean norm 
$||\vec{\mathrm{x}}||=\sqrt{\sum_{i=0}^{n-1} x_i^2}$ and $||\vec{\mathrm{x}}||_{\infty}=\mathrm{max}_i |x_i|$.
A matrix $\vec{\mathrm{B}}$ can be Gram-Schmidt orthogonalized in $\widetilde{\vec{\mathrm{B}}}$, and its norm
$||\vec{\mathrm{B}}||$ is the maximum of the norm of its columns. We denote by $(\vec{\mathrm{x}}|\vec{\mathrm{y}})$
the vector obtained as the concatenation of vectors $\vec{\mathrm{x}},\vec{\mathrm{y}}$. 
Let $\vec{\mathrm{I}}$ be the identity matrix 
and we denote by $\ln$ the neperian logarithm and $\log$ the binary logarithm. A lattice is the set of all integer 
linear combinations 
$\mathrm{\Lambda}(\vec{\mathrm{b}}_1,\ldots ,\vec{\mathrm{b}}_n)=\sum_i \vec{\mathrm{b_i}}\cdot x_i$ (where 
$x_i \in \mathbb{Z}$) of a set of linearly independent vectors $\vec{\mathrm{b}}_1,\ldots ,\vec{\mathrm{b}}_n$ 
called the basis of the lattice. If $\vec{\mathrm{B}}=[\vec{\mathrm{b}}_1,\ldots ,\vec{\mathrm{b}}_n]$ is the matrix 
basis, lattice vectors can be written as $\vec{\mathrm{B}}\vec{\mathrm{x}}$ for $\vec{\mathrm{x}}\in \mathbb{Z}^n$. 
Its dual $\mathrm{\Lambda}^*$ is the set of $\vec{\mathrm{x}}\in \R^n$ such that $\langle \vec{\mathrm{x}} , \mathrm{\Lambda} \rangle \subset \Z^n$. We have $\mathrm{\Lambda}^{**}=\mathrm{\Lambda}$.
We borrow Bleichenbacher's definition of bias \cite{de2013using}.
\begin{definition}
	The bias of a probability distribution $\phi$ over $\Zq$ is 
	\[ \E_{x\sim \phi}[\exp(2i\pi x/q)]. \]
\end{definition}
\noindent
This definition extends the usual definition of the bias of a coin in $\Z/2\Z$: it preserves the fact that any distribution with bias $b$ can be distinguished from uniform with constant probability using $\mathrm{\Omega}(1/b^{2})$ samples, as a consequence of Hoeffding's inequality; moreover the bias of the sum of two independent variable is still the product of their biases. We 
also have the following simple lemma:
\begin{lemma}\label{lemma:gaussianbias}
	The bias of the Gaussian distribution of mean $0$ and standard deviation $q\alpha$ is $\exp(-2\pi^2 \alpha^2)$.
\end{lemma}
\begin{proof}
The bias is the value of the Fourier transform at $-1/q$.
\end{proof}

We introduce a non standard definition for the $\mathsf{LWE}$ problem. However as a consequence of \autoref{lemma:gaussianbias}, this new definition naturally extends the usual Gaussian case (as well as its standard extensions such as the bounded noise variant \cite[Definition 2.14]{brakerski2013classical}), 
and it will prove easier to work with. The reader can consider the distorsion parameter $\epsilon=0$ as it is the case in other papers and a gaussian of standard deviation $\alpha q$. 

\begin{definition}
	Let $n\geq 0$ and $q \geq 2$ be integers.
	Given parameters $\alpha$ and $\epsilon$,
	the $\mathsf{LWE}$ distribution is, for $\vec{\mathrm{s}} \in (\Zq)^n$, a distribution on pairs $(\vec{\mathrm{a}},b)\in (\Zq)^n \times (\mathbb{R}/q\mathbb{Z})$ such that $\vec{\mathrm{a}}$ is sampled uniformly, and for all $\vec{\mathrm{a}}$,
	\[ |\E[\exp(2i\pi(\langle \vec{a} , \vec{s} \rangle -b)/q)|\vec{\mathrm{a}}]\exp(\alpha'^2)-1|\leq \epsilon \]
	for some universal $\alpha'\leq \alpha$.

	For convenience, we define $\beta=\sqrt{n/2}/\alpha$. In the remainder, $\alpha$ is called the \emph{noise} parameter\footnote{Remark that it differs by a constant factor from other authors' definition of $\alpha$.}, and $\epsilon$ the \emph{distortion} parameter.
	Also, we say that a $\mathsf{LWE}$ distribution has a noise distribution $\phi$ if $b$ is distributed as $\langle \vec{\mathrm{a}},\vec{\mathrm{s}}\rangle+\phi$.
\end{definition}

\begin{definition}
	The $\DecisionLWE$ problem is to distinguish a $\mathsf{LWE}$ distribution from the uniform distribution over $(\vec{\mathrm{a}},b)$.
	The $\SearchLWE$ problem is, given samples from a $\mathsf{LWE}$ distribution, to find $\vec{\mathrm{s}}$.
\end{definition}

\begin{definition}
	The real $\lambda_i$ is the radius of the smallest ball, centered in $\vec{0}$, such that it contains $i$ vectors of the 
	lattice $\mathrm{\Lambda}$ which are linearly independent.
\end{definition}

We define $\rho_s(\vec{\mathrm{x}})=\exp(-\pi ||\vec{\mathrm{x}}||^2/s^2)$ and $\rho_s(S)=\sum_{\vec{\mathrm{x}} \in S} \rho_s(\vec{\mathrm{x}})$ (and similarly for other functions).
The discrete Gaussian distribution $D_{E,s}$ over a set $E$ and of parameter $s$ is such that the probability of $D_{E,s}(\vec{\mathrm{x}})$ of drawing $\vec{\mathrm{x}} \in E$ is equal to $\rho_s(\vec{\mathrm{x}})/\rho_s(E)$. To simplify 
notation, we will denote by $D_E$ the distribution $D_{E,1}$.

\begin{definition}
	The smoothing parameter $\eta_\epsilon$ of the lattice $\mathrm{\Lambda}$ is the smallest $s$ such that $\rho_{1/s}(\mathrm{\Lambda}^*)=1+\epsilon$.
\end{definition}

Now, we will generalize the \textsf{BDD}, \textsf{UniqueSVP} and \textsf{GapSVP} problems by using another parameter 
$B$ that bounds the target lattice vector. For $B=2^n$, we recover the usual definitions if the input matrix is reduced. 
\begin{definition}
	The $\BDD_{B,\beta}^{||.||_\infty}$ $($resp. $\BDD_{B,\beta}^{||.||}$$)$ problem is, given a basis $\vec{\mathrm{A}}$ of the lattice $\mathrm{\Lambda}$, and a point $\vec{\mathrm{x}}$ such that $||\vec{\mathrm{As}}-\vec{\mathrm{x}}||\leq \lambda_1/\beta < \lambda_1/2$ and 
$||\vec{\mathrm{s}}||_{\infty}\leq B$ $($resp. $||\vec{\mathrm{s}}||\leq B$$)$, to find $\vec{\mathrm{s}}$.
\end{definition}

\begin{definition}
	The $\UniqueSVP_{B,\beta}^{||.||_\infty}$ $($resp. $\UniqueSVP_{B,\beta}^{||.||}$$)$ problem is, given a basis $\vec{\mathrm{A}}$ of the lattice $\mathrm{\Lambda}$, such that $\lambda_2/\lambda_1\geq \beta$ and there exists $\vec{\mathrm{s}}$ such that $||\vec{\mathrm{As}}||=\lambda_1$ with $||\vec{\mathrm{s}}||_{\infty}\leq B$ $($resp. $||\vec{\mathrm{s}}||\leq B$$)$, to find $\vec{\mathrm{s}}$.
\end{definition}

\begin{definition}
	The $\GapSVP_{B,\beta}^{||.||_\infty}$ $($resp. $\GapSVP_{B,\beta}^{||.||}$$)$ problem is, given a basis $\vec{\mathrm{A}}$ of the lattice $\mathrm{\Lambda}$ to 
	distinguish between $\lambda_1(\mathrm{\Lambda}) \geq \beta$ and if there exists $\vec{\mathrm{s}} \neq \vec{0}$ 
	such that $||\vec{\mathrm{s}}||_{\infty}\leq B$ $($resp. $||\vec{\mathrm{s}}||\leq B$$)$ and $||\vec{\mathrm{As}}||\leq 1$.
\end{definition}

\begin{definition}
	Given two probability distributions $P$ and $Q$ on a finite set $S$, the Kullback-Leibler (or $\mathsf{KL}$) divergence between $P$ and $Q$ is
$$D_{\mathsf{KL}}(P||Q)=\sum_{x\in S} \ln\bigg(\frac{P(x)}{Q(x)}\bigg)P(x) \ \mbox{	with }\ \ln(x/0)=+\infty \ \mbox{ if }\ x>0.$$
\end{definition}
The following two lemmata are proven in \cite{poppelmann2014enhanced} :
\begin{lemma}
	\label{KLmajore}
	Let $P$ and $Q$ be two distributions over $S$, such that for all $x$, $|P(x)-Q(x)|\leq \delta(x) P(x)$ with $\delta(x) \leq 1/4$.
	Then :
	\[ D_{\mathsf{KL}}(P||Q)\leq 2\sum_{x\in S}\delta(x)^2P(x). \]
\end{lemma}
\begin{lemma}
	\label{KLerreur}
	Let $A$ be an algorithm which takes as input $m$ samples of $S$ and outputs a bit.
	Let $x$ (resp. $y$) be the probability that it returns $1$ when the input is sampled from $P$ (resp. $Q$).
	Then :
	\[ |x-y|\leq \sqrt{mD_{\mathsf{KL}}(P||Q)/2}. \]
\end{lemma}

Finally, we say that an algorithm has a negligible probability of failure if its probability of failure is $2^{-\mathrm{\Omega}(n)}$. \footnote{Some authors use another definition.}

\subsection{Secret-Error Switching}
At a small cost in samples, it is possible to reduce any $\mathsf{LWE}$ distribution with noise distribution $\phi$ to an instance where the secret follows the rounded distribution, defined as $\lfloor \phi \rceil$~\cite{applebaum2009fast,brakerski2013classical}.

\begin{theorem}
\label{th:petitsecret}
	Given an oracle that solves $\mathsf{LWE}$ with $m$ samples in time $t$ with the secret coming from the rounded error distribution, it is possible to solve $\mathsf{LWE}$ with $m+\bigo(n\log \log q)$ samples with the same error distribution (and any distribution on the secret) in time $t + \bigo(mn^2+(n\log \log q)^3)$, with negligible probability of failure.

	Furthermore, if $q$ is prime, we lose $n+k$ samples with probability of failure bounded by $q^{-1-k}$.
\end{theorem}
\begin{proof}
	First, select an invertible matrix $\vec{\mathrm{A}}$ from the vectorial part of $\bigo(n\log \log q)$ samples in time $\bigo((n\log \log q)^3)$ \cite[Claim 2.13]{brakerski2013classical}.

	Let $\vec{b}$ be the corresponding rounded noisy dot products. Let $\vec{\mathrm{s}}$ be the $\mathsf{LWE}$ secret and 
	$\vec{\mathrm{e}}$ such that $\vec{\mathrm{As}}+\vec{\mathrm{e}}=\vec{\mathrm{b}}$.
	Then the subsequent $m$ samples are transformed in the following way.
	For each new sample $(\vec{\mathrm{a'}},b')$ with $b'=\langle \vec{\mathrm{a'}},\vec{\mathrm{s}} \rangle + e'$, we give the sample $(-^{t}\vec{\mathrm{A}}^{-1}\vec{\mathrm{a'}},b'-\langle ^{t}\vec{\mathrm{A}}^{-1}\vec{\mathrm{a'}}, \vec{\mathrm{b}} \rangle)$ to our $\mathsf{LWE}$ oracle.

	Clearly, the vectorial part of the new samples remains uniform and since
	\[ b'-\langle ^{t}\vec{\mathrm{A}}^{-1}\vec{\mathrm{a'}} , \vec{\mathrm{b}} \rangle = \langle -^{t}\vec{\mathrm{A}}^{-1}\vec{\mathrm{a'}} , \vec{\mathrm{b}}-\vec{\mathrm{As}} \rangle + b'-\langle \vec{\mathrm{a'}} ,\vec{\mathrm{s}} \rangle = \langle -^{t}\vec{\mathrm{A}}^{-1}\vec{\mathrm{a'}} , \vec{\mathrm{e}} \rangle + e'\]
	the new errors follow the same distribution as the original, and the new secret is $\vec{\mathrm{e}}$. Hence the oracle outputs $\vec{\mathrm{e}}$ in time $t$, and we can recover $\vec{\mathrm{s}}$ as $\vec{\mathrm{s}}=\vec{\mathrm{A}}^{-1}(\vec{\mathrm{b}}-\vec{\mathrm{e}})$.

	If $q$ is prime, the probability that the $n+k$ first samples are in some hyperplane is bounded by $q^{n-1}q^{-n-k}=q^{-1-k}$.
\end{proof}

\subsection{Low dimension algorithms}

Our main algorithm will return samples from a \textsf{LWE} distribution, while the bias decreases.
We describe two fast algorithms when the dimension is small enough.

\begin{theorem}
	If $n=0$ and $m=k/b^2$, with $b$ smaller than the real part of the bias, the $\DecisionLWE$ problem can be solved with advantage $1-2^{-\mathrm{\Omega}(k)}$ in time $\bigo(m)$.
\end{theorem}
\begin{proof}\label{th:distinguish}
	The algorithm {\sc Distinguish} computes $x=\frac{1}{m}\sum_{i=0}^{m-1} \cos(2i\pi b_i/q)$ and returns the boolean $x\geq b/2$.
	If we have a uniform distribution then the average of $x$ is $0$, else it is larger than $b/2$.
	The Hoeffding inequality shows that the probability of $|x-\E[x]|\geq b/2$ is $2^{-k/8}$, which gives the result.
\end{proof}

\begin{algorithm}
	\caption{FindSecret}
	\begin{algorithmic}
		\Function{FindSecret}{$\Li$}
		\ForAll{$(\vec{\mathrm{a}},b) \in \Li$}
		\State $f[\vec{\mathrm{a}}] \gets f[\vec{\mathrm{a}}] + \exp(2i\pi b/q)$
		\EndFor
		\State $t \gets \Call{FastFourierTransform}{f} $
		\State \Return $\argmax_{\vec{\mathrm{s}} \in (\Zq)^n} \Re(t[\vec{\mathrm{s}}])$
		\EndFunction

	\end{algorithmic}
\end{algorithm}

\begin{lemma}
	For all $\vec{\mathrm{s}}\neq \vec{0}$, if $\vec{\mathrm{a}}$ is sampled uniformly, $\E[\exp(2 i \pi \langle \vec{\mathrm{a}} , \vec{\mathrm{s}}\rangle /q)]=0$.
\end{lemma}
\begin{proof}
	Multiplication by $s_0$ in $\Z_q$ is $\gcd(s_0,q)$-to-one because it is a group morphism, therefore $a_0s_0$ is uniform over $\gcd(s_0,q)\Z_q$.
	Thus, by using $k=\gcd(q,s_0,\dots,s_{n-1}) < q$, $\langle \vec{\mathrm{a}} , \vec{\mathrm{s}} \rangle$ is distributed uniformly over $k\Z_q$ so
	\[ \E[\exp(2 i \pi \langle \vec{\mathrm{a}} , \vec{\mathrm{s}} \rangle /q)]=\frac{q}{k}\sum_{j=0}^{q/k-1} \exp(2i\pi jk/q)=0. \qedhere\]
\end{proof}

\begin{theorem}\label{th:findsecret}
	The algorithm {\sc FindSecret}, when given $m>(8n\log q+k)/b^2$ samples from a $\mathsf{LWE}$ problem with 
	bias whose real part is superior to $b$ returns the correct secret in time $\bigo(m+n\log^2(q)q^n)$ except with probability 
	$2^{-\mathrm{\Omega}(k)}$.
\end{theorem}
\begin{proof}
	The fast Fourier transform needs $\bigo(nq^n)$ operations on numbers of bit size $\bigo(\log(q))$.
	The Hoeffding inequality shows that the difference between $t[\vec{\mathrm{s'}}]$ and $\E[\exp(2i\pi (b-\langle \vec{\mathrm{a}} , \vec{\mathrm{s'}} \rangle)/q)]$ is at most $b/2$ except with probability at most $2\exp(-mb^2/2)$.
	It holds for all $\vec{\mathrm{s'}}$ except with probability at most 
	$2q^n\exp(-mb^2/2)=2^{-\mathrm{\Omega}(k)}$ using the union bound.
	Then $t[\vec{\mathrm{s}}]\geq b-b/2=b/2$ and for all $\vec{\mathrm{s'}} \neq \vec{\mathrm{s}}$, $t[\vec{\mathrm{s'}}] < b/2$ so the algorithm returns $\vec{\mathrm{s}}$.
\end{proof}

\section{Main algorithm}


In this section, we present our main algorithm, prove its asymptotical complexity, and present practical results in dimension $n=128$.

\subsection{Rationale}

A natural idea in order to distinguish between an instance of $\mathsf{LWE}$ (or $\mathsf{LPN}$) and a uniform 
distribution is to select some $k$ samples that add up to zero, yielding a new sample of the form $(\vec{0},e)$. It is then enough to distinguish between $e$ and a uniform variable. However, 
if $\delta$ is the bias of the error in the original samples, the new error $e$ has bias $\delta^{k}$, 
hence roughly $\delta^{-2k}$ samples are necessary to distinguish it from uniform. Thus it is crucial that $k$ be as small a possible.

The idea of the algorithm by Blum, Kalai and Wasserman BKW is to perform ``blockwise'' Gaussian elimination. 
The $n$ coordinates are divided into $k$ blocks of length $b = n/k$. Then, samples that are equal on the first $b$ 
coordinates are substracted together to produce new samples that are zero on the first block. This process is iterated over 
each consecutive block. Eventually samples of the form $(\vec{0},e)$ are obtained.

Each of these samples ultimately results from the addition of $2^{k}$ starting samples, so $k$ should be at most $\bigo(\log(n))$ for the algorithm to make sense. On the other hand $\mathrm{\Omega}(q^{b})$ data are clearly required at each step in order to generate enough collisions on $b$ consecutive coordinates of a block. This naturally results in a complexity roughly $2^{(1 + o(1))n/\log(n)}$ in the original algorithm for $\mathsf{LPN}$. This algorithm was later adapted to 
$\mathsf{LWE}$ in~\cite{albrecht2013complexity}, and then improved in~\cite{albrecht2014lazy}.

The idea of the latter improvement is to use so-called ``lazy modulus switching''. Instead of finding two vectors that are equal on a given block in order to generate a new vector that is zero on the block, one uses vectors that are merely close to each other. This may be seen as performing addition modulo $p$ instead of $q$ for some $p < q$, by rounding every value $x \in \Z_{q}$ to the value nearest $xp/q$ in $\Z_{p}$. Thus at each step of the algorithm, instead of generating vectors that are zero on each block, small vectors are produced. This introduces a new ``rounding'' error term, but essentially reduces the complexity from roughly $q^{b}$ to $p^{b}$. Balancing the new error term with this decrease in complexity results in a significant improvement.

However it may be observed that this rounding error is much more costly for the first few blocks than the last ones. Indeed samples produced after, say, one iteration step are bound to be added together $2^{a-1}$ times to yield the final samples, resulting in a corresponding blowup of the rounding error. By contrast, later terms will undergo less additions. Thus it makes sense to allow for progressively coarser approximations (i.e. decreasing the modulus) at each step. On the other hand, to maintain comparable data requirements to find collisions on each block, the decrease in modulus is compensated by progressively longer blocks.

What we propose here is a more general view of the \textsf{BKW} algorithm that allows for this improvement, while giving a clear view of the different complexity costs incurred by various choice of parameters. Balancing these terms is the key to finding an optimal complexity. We forego the ``modulus switching'' point of view entirely, while retaining its core ideas. The resulting algorithm generalizes several variants of \textsf{BKW}, and will be later applied in a variety of settings.

Also, each time we combine two samples, we never use again these two samples so that the combined samples are independent.
Previous works used repeatedly one of the two samples, so that independency can only be attained by repeating the entire algorithm for each sample needed by the distinguisher, as was done in~\cite{BKW}.

\subsection{Quantization}

The goal of quantization is to associate to each point of $\R^k$ a center from a \textit{small} set, such that the expectancy of the distance between a point and its center is small. We will then be able to produce small vectors by substracting vectors associated to the same center.

Modulus switching amounts to a simple quantizer which rounds every coordinate to the nearest multiple of some constant. Our proven algorithm uses a similar quantizer, except the constant depends on the index of the coordinate.

It is possible to decrease the average distance from a point to its center by a constant factor for large moduli \cite{gray1998quantization}, but doing so would complicate our proof without improving the leading term of the complexity. When the modulus is small, it might be worthwhile to use error-correcting codes as in \cite{guo2014solving}.

\subsection{Main Algorithm}

Let us denote by $\Li_{0}$ the set of starting samples, and $\Li_i$ the sample list after $i$ reduction steps. The numbers 
$d_{0} = 0 \leq d_{1} \leq \dots \leq d_{k} = n$ partition the $n$ coordinates of sample vectors into $k$ buckets. Let 
$\vec{\mathrm{D}} = (D_{0},\dots,D_{k-1})$ be the vector of quantization coefficients associated to each bucket.

\begin{algorithm}
	\caption{Main resolution}
	\begin{algorithmic}[1]
		\Function{Reduce}{$\Li_{in}$,$D_i$,$d_{i}$,$d_{i+1}$}
		\State $\Li_{out} \gets \varnothing $
		\State $t[] \gets \varnothing $
			\ForAll{$(\vec{\mathrm{a}},b) \in \Li_{in}$}
			\State $\vec{\mathrm{r}}=\lfloor \frac{(\vec{a}_{d_i},\dots,\vec{a}_{d_{i+1}-1})}{D} \rceil $
			\If {$t[\vec{\mathrm{r}}]=\varnothing$}
				\State $t[\vec{\mathrm{r}}] \gets (\vec{\mathrm{a}},b)$
			\Else
				\State $\Li_{out} \gets \Li_{out} :: \{(\vec{\mathrm{a,b}})-t[\vec{r}]\}$
				\State $t[\vec{\mathrm{r}}]\gets \varnothing$
			\EndIf
		\EndFor
		\State \Return $\Li_{out}$
		\EndFunction

		\Function{Solve}{$\Li_0$,$\vec{\mathrm{D}}$,$(d_{i})$}
		\For{$0\leq i<k$}
			\State $\Li_{i+1} \gets \Call{Reduce}{\Li_i,D_i,d_i,d_{i+1}}$
		\EndFor
		\State \Return $\Call{Distinguish}{\{b| (\vec{\mathrm{a}},b) \in \Li_k\}}$
		\EndFunction
	\end{algorithmic}
\end{algorithm}

In order to allow for a uniform presentation of the \textsf{BKW} algorithm, applicable to different settings, we do not assume a specific distribution on the secret. Instead, we assume there exists some {\it known} $\vec{\mathrm{B}} = (B_{0},\dots,B_{n-1})$ such that $\sum_i (s_i/B_i)^2 \leq n$. Note that this is in particular true if $|s_{i}| \leq B_{i}$. We shall see how to adapt this to the standard Gaussian case later on. Without loss of generality, $\vec{\mathrm{B}}$ is non increasing. 

There are $a$ phases in our reduction : in the $i$-th phase, the coordinates from $d_i$ to $d_{i+1}$ are reduced.
We define $m=|\Li_0|$.

\begin{lemma} {\sc Solve} terminates in time $\bigo(mn\log q)$.
\end{lemma}
\begin{proof}
	The {\sc Reduce} algorithm clearly runs in time $\bigo(|\Li|n \log q)$. Moreover, $|\Li_{i+1}|\leq |\Li_i|/2$ so that the total running time of {\sc Solve} is $\bigo(n\log q\sum_{i=0}^k m/2^i)=\bigo(mn\log q)$.
\end{proof}

\begin{lemma}
	\label{lemma:calculbruit}
	Write $\Li'_{i}$ for the samples of $\Li_{i}$ where the first $d_{i}$ coordinates of each sample vector have been truncated.
Assume $|s_j|D_{i}<0.23q$ for all $d_{i} \leq j < d_{i+1}$.
If $\Li'_{i}$ is sampled according to the $\mathsf{LWE}$ distribution of secret $\vec{\mathrm{s}}$ and noise parameters $\alpha$ and $\epsilon\leq 1$, then $\Li'_{i+1}$ is sampled according to the $\mathsf{LWE}$ distribution of the truncated secret 
with parameters:
\[\alpha'^2=2\alpha^2+4\pi^2\sum_{j=d_i}^{d_{i+1}-1}(s_jD_{i}/q)^2\quad\text{\rm and }\quad\epsilon'=3\epsilon.\]
	On the other hand, if $D_i=1$, then $\alpha'^2=2\alpha^2$.
\end{lemma}
\begin{proof}
	The independence of the outputted samples and the uniformity of their vectorial part are clear.
	Let $(\vec{\mathrm{a}},b)$ be a sample obtained by substracting two samples from $\Li_{i}$. For $\vec{\mathrm{a'}}$ the vectorial part of a sample, define $\epsilon(\vec{\mathrm{a'}})$ such that $\E[\exp(2i\pi(\langle \vec{\mathrm{a'}},\vec{\mathrm{s}} \rangle-b')/q)|\vec{\mathrm{a'}}]=(1+\epsilon(\vec{\mathrm{a'}}))\exp(-\alpha^2)$. By definition of \textsf{LWE}, $|\epsilon(\vec{\mathrm{a'}})| \leq \epsilon$, and by independence:
	\[ \E[\exp(2i\pi(\langle \vec{\mathrm{a}},\vec{\mathrm{s}}\rangle -b)/q)|\vec{\mathrm{a}}]=\exp(-2\alpha^2)\E_{\vec{\mathrm{a'}}-\vec{\mathrm{a''}}=\vec{\mathrm{a}}}[(1+\epsilon(\vec{\mathrm{a'}}))(1+\epsilon(\vec{\mathrm{a''}}))],\]
	with $|\E_{\vec{\mathrm{a'}}-\vec{\mathrm{a''}}=\vec{\mathrm{a}}}[(1+\epsilon(\vec{\mathrm{a'}}))(1+\epsilon(\vec{\mathrm{a''}}))]-1|\leq 3\epsilon$.
	
	\noindent
	Thus we computed the noise corresponding to adding two samples of $\Li_{i}$. To get the noise for a sample from $\Li_{i+1}$, it remains to truncate coordinates from $d_{i}$ to $d_{i+1}$.
	A straightforward induction on the coordinates shows that this noise is :
	\[ \exp(-2\alpha^2)\E_{\vec{\mathrm{a'}}-\vec{\mathrm{a''}}=\vec{\mathrm{a}}}[(1+\epsilon(\vec{\mathrm{a'}}))(1+\epsilon(\vec{\mathrm{a''}}))]\prod_{j=d_{i}}^{d_{i+1}-1}\E[\exp(2i\pi \vec{a}_j\vec{s}_j/q)].\]
	Indeed, if we denote by $\vec{\mathrm{a}}^{(j)}$ the vector $\vec{\mathrm{a}}$ where the first $j$ coordinates are truncated and $\alpha_j$ the noise parameter of $\vec{\mathrm{a}}^{(j)}$, we have:
	\begin{align*}
		& |\E[\exp(2i\pi(\langle \vec{\mathrm{a}}^{(j+1)},\vec{\mathrm{s}}^{(j+1)} \rangle-b)/q)|\vec{\mathrm{a}}^{(j+1)}]-\exp(-\alpha_n^2)\E[\exp(2i\pi \vec{a}_j\vec{s}_j/q)]|\\
		=\;&|\E[\exp(-2i\pi \vec{a}_j\vec{s}_j/q)(\exp(2i\pi(\langle \vec{\mathrm{a}}^{(j)},\vec{\mathrm{s}}^{(j)} \rangle -b)/q)-\exp(-\alpha_j^2))]|\\
		\leq\;&\epsilon' \exp(-\alpha_j^2)\E[\exp(2i\pi \vec{a}_j\vec{s}_j/q)].
	\end{align*}

\noindent
	It remains to compute $\E[\exp(2i\pi \vec{a}_j\vec{s}_j/q)]$ for $d_{i} \leq j<d_{i+1}$.
	Let $D = D_{i}$. The probability mass function of $\vec{a}_j$ is even, so $\E[\exp(2i\pi \vec{a}_j \vec{s}_j/q)]$ is real.
	Furthermore, since $|\vec{a}_j|\leq D$, \[ \E[\exp(2i\pi \vec{a}_j \vec{s}_j/q)]\geq \cos(2\pi \vec{s}_jD/q) .\]
	Simple function analysis shows that $\ln(\cos(2\pi x))\geq -4\pi^2x^2$ for $|x|\leq 0.23$, and since $|\vec{s}_j|D<0.23q$, we get :
	\[ \E[\exp(2i\pi \vec{a}_j \vec{s}_j/q)]\geq \exp(-4\pi^2 \vec{s}_j^2D^2/q^2). \]
	On the other hand, if $D_i=1$ then $\vec{a}_j=0$ and $\E[\exp(2i\pi \vec{a}_j \vec{s}_j/q)]=1$.
\end{proof}

Finding optimal parameters for \textsf{BKW} amounts to balancing various costs: the baseline number of samples required so that the final list $\Li_{k}$ is non-empty, and the additional factor due to the need to distinguish the final error bias. This final bias itself comes both from the blowup of the original error bias by the \textsf{BKW} additions, and the ``rounding errors'' due to quantization. Balancing these costs essentially means solving a system.

For this purpose, it is convenient to set the overall target complexity as $2^{n(x + o(1))}$ for some $x$ to be determined. The following auxiliary lemma essentially gives optimal values for the parameters of {\sc Solve} assuming a suitable value of $x$. The actual value of $x$ will be decided later on.

\begin{lemma}\label{lemma:main}
	Pick some value $x$ (dependent on $\mathsf{LWE}$ parameters).
	Choose:
	\begin{align*}
	k &\leq \bigg\lfloor \log\bigg(\frac{nx}{6\alpha^2}\bigg) \bigg\rfloor\quad&
	m &= n2^k2^{nx}\\
	D_i &\leq \frac{q\sqrt{x/6}}{\pi B_{d_i}2^{(a-i+1)/2}}\quad&
	d_{i+1} &= \min\bigg(d_{i} + \bigg\lfloor \frac{nx}{\log(1+q/D_i)} \bigg\rfloor, n\bigg).
	\end{align*}
	Assume $d_k = n$ and $\epsilon\leq 1/(\beta^{2}x)^{\log 3}$, and for all $i$ and $d_i\leq j < d_{i+1}$, $|s_j|D_i<0.23q$. {\sc Solve} runs in time $\bigo(mn)$ with negligible failure probability.
\end{lemma}
\begin{proof}
	Remark that for all $i$, \[ |\Li_{i+1}|\geq (|\Li_{i}|-(1+q/D_i)^{d_{i+1}-d_i})/2 \geq (|\Li_i|-2^{nx})/2. \]
	Using induction, we then have $|\Li_i|\geq (|\Li_0|+2^{nx})/2^i-2^{nx}$ so that $|\Li_k| \geq n2^{nx}$.

	By induction and using the previous lemma, the input of {\sc Distinguish} is sampled from a $\mathsf{LWE}$ distribution with noise parameter:
	\[ \alpha'^2=2^k\alpha^2+4\pi^2\sum_{i=0}^{k-1}2^{k-i-1}\sum_{j=d_i}^{d_{i+1}-1}(s_jD_i/q)^2. \]
	By choice of $k$ the first term is smaller than $nx/6$.
	As for the second term, since $B$ is non increasing and by choice of $D_{i}$, it is smaller than:
	\[ 4\pi^2\sum_{i=0}^{k-1}2^{k-i-1}\frac{x/6}{\pi^22^{k-i+1}}\sum_{j=d_i}^{d_{i+1}-1}\Big(\frac{s_j}{B_j}\Big)^2
		\leq (x/6)\sum_{j=0}^{n-1}\Big(\frac{s_{j}}{B_{j}}\Big)^{2}\leq nx/6. \]
	Thus the real part of the bias is superior to $\exp(-nx/3)(1-3^a\epsilon) \geq 2^{-nx/2}$, and hence by \autoref{th:distinguish}, {\sc Distinguish} fails with negligible probability.
\end{proof}
\begin{theorem}\label{thm:lwe}
	Assume that for all $i$, $|s_i|\leq B$, $B\geq 2$, $\max(\beta,\log(q))=2^{o(n/\log n)}$, $\beta=\omega(1)$,  and $\epsilon\leq 1/\beta^4$.
	Then {\sc Solve} takes time $2^{(n/2+o(n))/\ln(1+\log \beta/\log B))}$.
\end{theorem}
\begin{proof}
	We apply \autoref{lemma:main}, choosing \[ k=\lfloor \log(\beta^2/(12\ln(1+\log \beta))) \rfloor=(2-o(1))\log \beta \in \omega(1)\] and we set $D_i=q/(Bk2^{(k-i)/2})$. It now remains to show that this choice of parameters satisfies the conditions of the lemma.

	First, observe that $BD_i/q\leq 1/k=o(1)$ so the condition $|s_j|D_i<0.23q$ is fulfilled.
	Then,  $d_k \geq n$, which amounts to:
	\[ \sum_{i=0}^{k-1} \frac{x}{(k-i)/2+\log \bigo(kB)} \geq 2x\ln(1+k/2/\log \bigo(kB)) \geq 1+k/n=1+o(1) \]

	If we have $\log k=\omega(\log\log B)$ (so in particular $k = \omega(\log B)$), we get $\ln(1+k/2/\log \bigo(kB))=(1+o(1))\ln(k)=(1+o(1))\ln(1+\log \beta/\log B)$.

	Else, $\log k=\bigo(\log \log B)=o(\log B)$ (since necessarily $B = \omega(1)$ in this case), so we get $\ln(1+k/2/\log \bigo(kB))=(1+o(1))\ln(1+\log \beta/\log B)$.
%

	Thus our choice of $x$ fits both cases and we have $1/x\leq 2\ln(1+\log \beta)$.
	Second, we have $1/k=o(\sqrt{x})$ so $D_i$, $\epsilon$ and $k$ are also sufficiently small and the lemma applies.
	Finally, note that the algorithm has complexity $2^{\Omega(n/\log n)}$, so a factor $n2^k\log(q)$ is negligible.
\end{proof}

This theorem can be improved when the use of the given parameters yields $D<1$, since $D=1$ already gives a lossless quantization.

\begin{theorem}
	\label{thm:lwepetitmodulo}
	Assume that for all $i$, $|s_i|\leq B=n^{b+o(1)}$.
	Let $\beta=n^{c}$ and $q=n^d$ with $d\geq b$ and $c+b\geq d$. Assume $\epsilon\leq 1/\beta^4$.
	Then {\sc Solve} takes time $2^{n/(2(c-d+b)/d+2\ln(d/b)-o(1))}$.
\end{theorem}
\begin{proof}
	Once again we aim to apply \autoref{lemma:main}, and choose $k$ as above: \[ k=\log(\beta^2/(12\ln(1+\log \beta)))=(2c-o(1))\log n \]
	If $i<\lceil 2(c-d+b)\log n \rceil$, we take $D_i=1$, else we choose $q/D_i=\Theta(B2^{(a-i)/2})$.
	Satisfying $d_{a} \geq n-1$ amounts to:
	\begin{align*}
		& 2x(c-d+b)\log n/\log q+\sum_{i=\lceil 2(c-d+b)\log n \rceil}^{a-1} \frac{x}{(a-i)/2+\log \bigo(B)} \\
		\geq\;&2x(c-d+b)/d+2x\ln((a-2(c-d+b)\log n+2\log B)/2/\log \bigo(B)) \\
		\geq\;&1+a/n=1+o(1) 
	\end{align*}
	So that we can choose $1/x=2(c-d+b)/d+2\ln(d/b)-o(1)$.\qedhere
\end{proof}

\begin{corollary}
Given a $\mathsf{LWE}$ problem with $q=n^d$, Gaussian errors with $\beta=n^c$, $c>1/2$ and $\epsilon\leq n^{-4c}$, we can 
find a solution in $2^{n/(1/d+2\ln(d/(1/2+d-c))-o(1))}$ time. 
\end{corollary}
\begin{proof}
	Apply \autoref{th:petitsecret} : with probability $2/3$, the secret is now bounded by $B=\bigo(q\sqrt{n}/\beta \sqrt{\log n})$.
	The previous theorem gives the complexity of an algorithm discovering the secret, using $b=1/2-c+d$, and which works with probability $2/3-2^{-\Omega(n)}$.
	Repeating $n$ times with different samples, the correct secret will be outputted at least $n/2+1$ times, except with negligible probability.
	By returning the most frequent secret, the probability of failure is negligible.
\end{proof}
\noindent
In particular, if $c \leq d$, it is possible to quantumly approximate lattice problems within factor $\bigo(n^{c+1/2})$~\cite{regev2009lattices}.
Setting $c=d$, the complexity is $2^{n/(1/c+2\ln(2c)-o(1))}$, so that the constant slowly converges to $0$ when $c$ goes to infinity.

A simple $\mathsf{BKW}$ using the bias would have a complexity of $2^{d/cn+o(n)}$, the analysis of \cite{albrecht2014lazy} or \cite{albrecht2013complexity} only conjectures $2^{dn/(c-1/2)+o(n)}$ for $c>1/2$.
In \cite{albrecht2014lazy}, the authors incorrectly claim a complexity of $2^{cn+o(n)}$ when $c=d$, because the blowup in the error is not explicitely computed~\footnote{They claim it is possible to have $\approx\log n$ reduction steps while the optimal number is $\approx 2c \log(n)$ so that is loose for $c>1/2$ and wrong for $c<1/2$.}.

Though the upper bound on $\epsilon$ is small, provably solving the learning with rounding~\cite{banerjee2012pseudorandom} problem where $b=\frac{q}{p}\lfloor \frac{p}{q} \langle \vec{a},\vec{s} \rangle \rceil$ for some $p$ seems out of reach\footnote{\cite{duc2015better} claims to do so, but it actually assumes the indepedency of the errors from the vectorial part of the samples.}. 

Finally, if we want to solve the $\mathsf{LWE}$ problem for different secrets but with the same vectorial part of the samples, it is possible to be much faster if we work with a bigger final bias, since the {\sc Reduce} part needs to be called only once.

\subsection{Some Practical Improvements}\label{app:heuristics}

In this section we propose a few heuristic improvements for our main algorithm, which speed it up in practice, although they do not change the factor in the exponent of the overall complexity. In our main algorithm, after a few iterations of {\sc Reduce}, each sample is a sum of a few original samples, and these sums are disjoint. It follows that samples are independent. This may no longer be true below, and sample independence is lost, hence the heuristic aspect. However this has negligible impact in practice.

First, in the non-binary case, when quantizing a sample to associate it to a center, we can freely quantize its opposite as well (i.e. quantize the sample made up of opposite values on each coordinate) as in~\cite{albrecht2013complexity}.

Second, at each {\sc Reduce} step, instead of substracting any two samples that are quantized to the same center, we could choose samples whose difference is as small as possible, among the list of samples that are quantized to the same center. The simplest way to achieve this is to generate the list of all differences, and pick the smallest elements (using the L2 norm). We can thus form a new sample list, and we are free to make it as long as the original. Thus we get smaller samples overall with no additional data requirement, at the cost of losing sample independence.\footnote{A similar approach was taken in~\cite{albrecht2014lazy} where the L1 norm was used, and where each new sample is reduced with the shortest element ever quantized on the same center.}

Analyzing the gain obtained using this tactic is somewhat tricky, but quite important for practical optimization. One approach is to model reduced sample coordinates as independent and following the same Gaussian distribution. When adding vectors, the coordinates of the sum then also follows a Gaussian distribution. The (squared) norm of a vector with Gaussian coordinates follows the $\chi^{2}$ distribution. Its cumulative distribution for a $k$-dimensional vector is the regularized gamma function, which amounts to $1-\exp(-x/2)\sum_{i=0}^{k/2}(x/2)^{i}/i!$ for even $k$ (assuming standard deviation 1 for sample coordinates).

Now suppose we want to keep a fixed proportion of all possible sums. Then using the previous formula we can compute the radius $R$ such that the expected number of sample sums falling within the ball of radius $R$ is equal to the desired proportion. Thus using this Gaussian distribution model, we are able to predict how much our selection technique is able to decrease the norm of the samples.

Practical experiments show that for a reasonable choice of parameters (namely keeping a proportion 1/1000 of samples for instance), the standard deviation of sample norms is about twice as high as predicted by the Gaussian model for the first few iterations of {\sc Reduce}, then falls down to around $15\%$ larger. This is due to a variety of factors; it is clear that for the first few iterations, sample coordinates do not quite follow a Gaussian distribution. Another notable observation is that newly reduced coordinates are typically much larger than others.

While the previous Gaussian model is not completely accurate, the ability to predict the expected norms for the optimized algorithm is quite useful to optimize parameters. In fact we could recursively compute using dynamic programming what minimal norm is attainable for a given dimension $n$ and iteration count $a$ within some fixed data and time complexities. In the end we will gain a constant factor in the final bias.

In the binary case, we can proceed along a similar line by assuming that coordinates are independent and follow a Bernoulli distribution.

Regarding secret recovery, in practice, it is worthwhile to compute the Fourier transform on a high dimension, such that its cost approximately matches that of {\sc Solve}. On the other hand, for a low enough dimension, computing the experimental bias for each possible high probability secret may be faster.

Another significant improvement in practice can be to apply a linear quantization step just before secret recovery. The quantization steps we have considered are linear, in the sense that centers are a lattice. If $\vec{\mathrm{A}}$ is a basis of this lattice, in the end we are replacing a sample $(\vec{\mathrm{a}},b)$ by $(\vec{\mathrm{Ax}},b)$. We get $\langle \vec{\mathrm{x}} , \vec{\mathrm{A}}^t\vec{\mathrm{s}}\rangle +\vec{\mathrm{e}}= \langle\vec{\mathrm{Ax}}, \vec{\mathrm{s}}\rangle +\vec{\mathrm{e}} = b-\langle (\vec{\mathrm{a}}-\vec{\mathrm{Ax}}), \vec{\mathrm{s}}\rangle$. Thus the dimension of the Fourier transform is decreased as remarked by \cite{guo2014solving}, at the cost of a lower bias. Besides, we no longer recover $\vec{\mathrm{s}}$ but $\vec{\mathrm{y}}=\vec{\mathrm{A}}^t\vec{\mathrm{s}}$.
Of course we are free to change $\vec{\mathrm{A}}$ and quantize anew to recover more information on $\vec{\mathrm{s}}$. In some cases, if the secret is small, it may be worth simply looking for the small solutions of $\vec{\mathrm{A}}^t\vec{\mathrm{x}}=\vec{\mathrm{y}}$ (which may in general cost an exponential time) and test them 
against available samples.

In general, the fact that the secret is small can help its recovery through a maximum likelihood test \cite{neymanpearson}.
If the secret has been quantized as above however, $\vec{\mathrm{A}}$ will need to be chosen such that $\vec{\mathrm{As}}$ is small (by being sparse or having small entries).
The probability to get a given secret can then be evaluated by a Monte-Carlo approach with reasonable accuracy within negligible time compared to the Fourier transform.

If the secret has a coordinate with non-zero mean, it should be translated.

Finally, in the binary case, it can be worthwhile to combine both of the previous algorithms: after having reduced vectors, we can assume some coordinates of the secret are zero, and quantize the remainder. Depending on the number of available samples, a large number of secret-error switches may be possible. Under the assumption that the success probability is independent for each try, this could be another way to essentially proceed as if we had a larger amount of data than is actually available. We could thus hope to operate on a very limited amount of data (possibly linear in $n$) for a constant success probability.

The optimizations above are not believed to change the asymptotic complexity of the algorithm, but have a significant impact in practice.

\begin{table}
\centering
\caption{Complexity of solving LWE with the Regev parameters~\cite{regev2009lattices}, i.e. the error distribution is a continuous gaussian of standard deviation $\frac{q}{\sqrt{2\pi n}\log^2(n)}$ and with modulus $q\approx n^2$. The reasonable column corresponds to multiplying the predicted reduction factor at each step by $1.1$, and assuming that the quantizers used reduce the variance by a factor of $1.3$. The corresponding parameters of the algorithm are shown in column $k$ (the number of reduction steps), $\log(m)$ ($m$ is the list size), and $\log(N)$ (one vector is kept for the next iteration for each $N$ vectors tested). The complexities are expressed as logarithm of the number of bit operations of each reduction step. Pessimistic uses a multiplier of $2$ and a naive quantifier (factor $1$). Optimistic uses a multiplier of $1$ and an asymptotical quantifier (factor $2\pi\me/12 \approx 1.42$). The asymptotical complexity is $2^{0.93n+o(n)}$ instead of $2^{n+o(n)}$.}
\begin{tabular}{| c | c | c | c | c | c | c | c | c|}
	\hline
	$n$ & $q$ & $k$ & $\log(m)$ & $\log(N)$ & Reasonable & Optimistic & Pessimistic & Previous~\cite{duc2015better} \\
	\hline
	64 & 4099 & 16 & 30 & 0 & 39.6 & 39.6 & 40.6 & 56.2 \\
	\hline
	80 & 6421 & 17 & 38 & 0 & 47.9 & 46.0 & 48.0 & 66.9 \\
	\hline
	96 & 9221 & 18 & 45 & 0 & 55.3 & 54.3 & 56.3 & 77.4 \\
	\hline
	112 & 12547 & 18 & 54 & 0 & 64.6 & 60.6 & 65.6 & 89.6 \\
	\hline
	128 & 16411 & 19 & 60 & 0 & 70.8 & 67.8 & 72.8 & 98.8 \\
	\hline
	160 & 25601 & 20 & 75 & 0 & 86.2 & 82.2 & 88.2 & 119.7 \\
	\hline
	224 & 50177 & 21 & 93 & 13 & 117.8 & 111.8 & 121.8 & 164.3 \\
	\hline
	256 & 65537 & 22 & 106 & 15 & 133.0 & 125.0 & 137.0 & 182.7 \\
	\hline
	384 & 147457 & 24 & 164 & 18 & 194.7 & 183.7 & 201.7 & 273.3 \\
	\hline
	512 & 262147 & 25 & 219 & 25 & 257.2 & 242.2 & 266.2 & 361.6 \\
	\hline
\end{tabular}
\end{table}

\begin{table}
\centering
\caption{
Complexity of solving LWE with the Lindner-Peikert parameters~\cite{DBLP:conf/ctrsa/LindnerP11} but with a number of samples $m$ much larger than the cryptosystem provides ($2n+128$). The error distribution is $D_{\Z,s}$.}

\begin{tabular}{| c | c | c | c | c | c | c | c | c|}
	\hline
	$n$ & $q$ & $s$ & $k$ & $\log(m)$ & $\log(N)$ & Reasonable & Optimistic & Pessimistic\\
	\hline
	192 & 4099 & 8.87 & 19 & 68 & 5 & 84.2 & 79.2 & 84.2 \\
	\hline
	256 & 6421 & 8.35 & 20 & 82 & 8 & 101.7 & 95.7 & 103.7 \\
	\hline
	320 & 9221 & 8.00 & 22 & 98 & 9 & 119.0 & 112.0 & 122.0 \\
	\hline
\end{tabular}
\end{table}

\begin{table}
\centering
\caption{Complexity of solving LWE with binary ($\zeroun$) secret with the Regev parameters~\cite{regev2009lattices}. }
\begin{tabular}{| c | c | c | c | c | c | c | c | c|}
	\hline
	$n$ & $q$ & $k$ & $\log(m)$ & $\log(N)$ & Reasonable & Optimistic & Pessimistic & Previous~\cite{albrecht2014lazy} \\
	\hline
	128 & 16411 & 16 & 28 & 0 & 38.8 & 38.8 & 39.8 & 74.2 \\
	\hline
	256 & 65537 & 19 & 52 & 0 & 64.0 & 62.0 & 67.0 & 132.5 \\
	\hline
	512 & 262147 & 22 & 99 & 0 & 112.2 & 104.2 & 117.2 & 241.8 \\
	\hline
\end{tabular}
\end{table}

\subsection{Experimentation}

We have implemented our algorithm, in order to test its efficiency in practice, as well as that of the practical improvements in \autoref{app:heuristics}. We have chosen dimension $n = 128$, modulus $q = n^{2}$, binary secret, and Gaussian errors with noise parameter $\alpha = 1/(\sqrt{n/\pi}\log^2 n)$. The previous best result for these parameters, using a $\mathsf{BKW}$ algorithm with lazy modulus switching, claims a time complexity of $2^{74}$ with $2^{60}$ samples \cite{albrecht2014lazy}.

Using our improved algorithm, we were able to recover the secret using $m = 2^{28}$ samples within 13 hours on a single PC equipped with a 16-core Intel Xeon. The computation time proved to be devoted mostly to the computation of $9\cdot 10^{13}$ norms, computed in fixed point over 16 bits in SIMD.
The implementation used a naive quantizer.

In~\autoref{app:LLLbinary}, we compare the different techniques to solve the \textsf{LWE} problem when the 
number of samples is large or small.
We were able to solve the same problem using BKZ with block size 40 followed by an enumeration in two minutes.

\subsection{Extension to the norm L2}
\label{app:secL2}
One can think that in~\autoref{lemma:calculbruit}, the condition $|s_j|D_i<0.23 q$ is artificial.
In this section, we show how to remove it, so that the only condition on the secret is $||s||\leq \sqrt{n}B$ which is always better.
The basic idea is to use rejection sampling to transform the coordinates of the vectorial part of the samples into small discrete gaussians.
The following randomized algorithm works only for some moduli, and is expected to take more time, though the asymptotical complexity is the same.

\begin{algorithm}
	\begin{algorithmic}[1]
		\Function{Accept}{$\vec{a}$,$\vec{u}$,$\sigma$,$k$} \Comment{For all $i$, $\vec{a}_i\in [0;k[$}
		\State \Return \textbf{true} with probability $\exp(-\pi(\sum_i \min(\vec{u}_i^2,(\vec{u}_i+1)^2)-(\vec{u_i}+\vec{a}_i/k)^2)/\sigma^2)$
		\EndFunction
		\Function{ReduceL2}{$\Li_{in}$,$D_i$,$d_{i}$,$d_{i+1}$,$\sigma_i$}
		\State $t[] \gets \varnothing $
		\ForAll{$(\vec{\mathrm{a}},b) \in \Li_{in}$}
			\State $\vec{\mathrm{r}}=\lfloor \frac{(\vec{a}_{d_i},\dots,\vec{a}_{d_{i+1}-1})}{D} \rceil $
			\State $\Call{Push}{t[\vec{\mathrm{r}}],(\vec{\mathrm{a}},b)}$
		\EndFor
		\State $\Li_{out} \gets \varnothing $
		\While{$|\{\vec{r}\in (\Z/(q/D_i)\Z)^{d_{i+1}-d_i};t[\vec{r}]\neq \varnothing\}|\geq (q/D_i)^{d_{i+1}-d_i}/3$}
			\State Sample $\vec{x}$ and $\vec{y}$ according to $D_{\Z^{d_{i+1}-d_i},\sigma_i}$
			\Repeat
				\State Sample $\vec{u}$ and $\vec{v}$ uniformly in $(\Z/(q/D_i)\Z)^{d_{i+1}-d_i}$
			\Until{$t[\vec{u}+\vec{x}]\neq \varnothing$ and $t[\vec{v}+\vec{y}]\neq \varnothing$}
			\State $(\vec{a_0},b_0)\gets \Call{Pop}{t[\vec{u}+\vec{x}]}$
			\State $(\vec{a_1},b_1)\gets \Call{Pop}{t[\vec{v}+\vec{y}]}$
		\If{$\Call{Accept}{\vec{a_0} \Mod q/D_i,\vec{u},\sigma_i,q/D_i}$ and $\Call{Accept}{\vec{a_1} \Mod q/D_i,\vec{v},\sigma_i,q/D_i}$}
				\State $\Li_{out} \gets (\vec{a_0}-\vec{a_1},b_0-b_1)::\Li_{out}$
			\EndIf
		\EndWhile
		\State \Return $\Li_{out}$
		\EndFunction
	\end{algorithmic}
\end{algorithm}

\begin{lemma}
	\label{lemma:calculbruitL2}
	Assume that $D_i|q$, $\sigma_i$ is larger than some constant and \[ |\Li_i|\geq 2n\max(n\log(q/D_i)(q/D_i)^{d_{i+1}-d_i},\exp(5(d_{i+1}-d_i)/\sigma_i)). \]
	Write $\Li'_{i}$ for the samples of $\Li_{i}$ where the first $d_{i}$ coordinates of each sample vector have been truncated.
If $\Li'_{i}$ is sampled according to the $\mathsf{LWE}$ distribution with secret $\vec{\mathrm{s}}$ and noise parameters $\alpha$ and $\epsilon$, then $\Li'_{i+1}$ is sampled according to the $\mathsf{LWE}$ distribution of the truncated secret 
with parameters:
\[\alpha'^2=2\alpha^2+2\pi\sum_{j=d_i}^{d_{i+1}-1}(s_j\sigma_iD_{i}/q)^2\quad\text{\rm and }\quad\epsilon'=3\epsilon.\]
	On the other hand, if $D_i=1$, then $\alpha'^2=2\alpha^2$.
	Furthermore, {\sc ReduceL2} runs in time $\bigo(n\log q|\Li_i|)$ and $|\Li_{i+1}| \geq |\Li_i|\exp(-5(d_{i+1}-d_i)/\sigma_i)/6$ except with probability $2^{-\Omega(n)}$.
\end{lemma}
\begin{proof}
	On lines 16 and 17, $\vec{a_0}\Mod q/D_i$ and $\vec{a_1} \Mod q/D_i$ are uniform and independent.
	On line 19, because of the rejection sampling theorem, we have $(q/D_i)\vec{u}+(\vec{a_0}\Mod q/D_i)$ and $(q/D_i)\vec{v}+(\vec{a_1}\Mod q/D_i)$ sampled according to $D_{\Z^{d_{i+1}-d_i},\sigma_i q/D_i}$.
	Also, the unconditional acceptance probability is, for sufficiently large $\sigma_i$,
	\begin{align*}
		(\rho_{\sigma_i q/D_i}(\Z)/(q/D_i(1+\rho_{\sigma_i}(\Z))))^{2(d_{i+1}-d_i)} \geq & \\
	 (\sigma_i q/D_i/(q/D_i(2+\sigma_i)))^{2(d_{i+1}-d_i)} \geq & \\
		\exp(-5(d_{i+1}-d_i)/\sigma_i) &
	\end{align*}
	where the first inequality comes from Poisson summation on both $\rho$.
	Next, the bias added by truncating the samples is the bias of the scalar product of $(q/D_i)\vec{u}+(\vec{a_0}\Mod q/D_i)-(q/D_i)\vec{v}-(\vec{a_1}\Mod q/D_i)$ with the corresponding secret coordinates.
	Therefore, using a Poisson summation, we can prove that $\alpha'$ is correct.

	The loop on line 15 is terminated with probability at least $2/3$ each time so that the Hoeffding inequality proves the complexity.

	On line 10, the Hoeffding inequality proves that $\min_{\vec{r}} t[\vec{r}]\geq |\Li_i|/(2(q/D_i)^{d_{i+1}-d_i})$ holds except with probability $2^{-\Omega(n)}$.
	Under this assumption, the loop on lines 11-21 is executed at least $|\Li_i|/6$ times.
	Then, the Hoeffding inequality combined with the bound on the unconditional acceptance probability proves the lower bound on $|\Li_{i+1}|$.
\end{proof}

\begin{theorem} \label{thm:lweL2}
	Assume that $||s||\leq \sqrt{n}B$, $B\geq 2$, $\max(\beta,\log(q))=2^{o(n/\log n)}$, $\beta=\omega(1)$,  and $\epsilon\leq 1/\beta^4$.
	Then, there exists an integer $Q$ smaller than $B^n2^{2n^2}$ such that if $q$ is divisible by this integer, we can solve $\DecisionLWE$ time \[2^{(n/2+o(n))/\ln(1+\log \beta/\log B)}. \]
\end{theorem}
\begin{proof}
	We use $\sigma_i=\sqrt{\log \beta}$, $m=2n^26^k\exp(5n/\sigma_i)2^{nx}$,\[ k=\lfloor \log(\beta^2/(12\ln(1+\log \beta)))\rfloor =(2-o(1))\log \beta \in \omega(1)\] and we set $D_i=q/(Bk2^{(k-i)/2})$, $d_{i+1}=\min(d_i+\lfloor \frac{nx}{\log(q/D_i)} \rfloor,n)$ and $Q=\prod_i Bk2^{(k-i)/2}$.
	We can then see as in~\autoref{thm:lwe} that these choices lead to some $x=1/(2-o(1))/\ln(1+\log \beta/\log B)$.
	Finally, note that the algorithm has complexity $2^{\Omega(n/\log \log \beta)}$, so a factor $2n^26^k\log(q)\exp(5n/\sqrt{\log(\beta)})$ is negligible.
\end{proof}
The condition on $q$ can be removed using modulus switching~\cite{brakerski2013classical}, unless $q$ is tiny.
\begin{theorem}
	If $B\beta\leq q$, then the previous theorem holds without the divisibility condition.
\end{theorem}
\begin{proof}
	Let $p\geq q$ be the smallest modulus such that the previous theorem applies and $\varsigma=\sqrt{n}p/q$.
	For each sample $(\vec{a},b)$, sample $\vec{x}$ from $D_{\Z^n-\vec{a}/q,\varsigma}$ and use the sample $(p/q\vec{a}+p\vec{x},p/qb)$ with the algorithm of the previous theorem.
	Clearly, the vectorial part has independent coordinates, and the probability that one coordinate is equal to $y$ is proportional to $\rho_\varsigma(p/q\Z-y)$.
	Therefore, the Kullback-Leibler distance with the uniform distribution is $2^{-\Omega(n)}$.
	If the original error distribution has a noise parameter $\alpha$, then the new error distribution has noise parameter $\alpha'$ with 
	\[ \alpha'^2\leq \alpha^2+\sum_i \pi\varsigma^2s_i^2/p^2 \leq \alpha^2+\pi B^2\varsigma^2/p^2 \leq n/2/\beta^2+\pi nB^2/q^2 \]
	so that $\beta$ is only reduced by a constant.
\end{proof}

\section{Applications to Lattice Problems}

We first show that $\BDD_{B,\beta}$ is easier than $\mathsf{LWE}$ for some large enough modulus
and then that $\UniqueSVP_{B,\beta}$ and $\GapSVP_{B,\beta}$ are easier than $\BDD_{B,\beta}$. 

\subsection{Variant of Bounding Distance Decoding}
The main result of this subsection is close to the classic reduction of \cite{regev2009lattices}.
However, our definition of $\mathsf{LWE}$ allows to simplify the proof, and gain a constant factor in the decoding radius.
The use of the KL divergence instead of the statistical distance also allows to gain a constant factor, when we 
need an exponential number of samples, or when $\lambda_n^*$ is really small.

The core of the reduction lies in Lemma~\ref{LemmeBDDreduction}, assuming access to a Gaussian sampling oracle. This hypothesis will be taken care of in Lemma~\ref{lemma:freesampling}.

\begin{lemma}
	\label{LemmeBDDreduction}
	Let $\vec{\mathrm{A}}$ be a basis of the lattice $\mathrm{\Lambda}$ of full rank $n$.
	Assume we are given access to an oracle outputting a vector sampled under the law $D_{\mathrm{\Lambda}^*,\sigma}$ and  $\sigma \geq q\eta_\epsilon(\mathrm{\Lambda}^*)$, and to an oracle solving the $\mathsf{LWE}$ problem in dimension $n$, modulus $q\geq 2$, noise parameter $\alpha$, and distortion parameter $\xi$ which fails with negligible probability and use $m$ vectors if the secret $\vec{\mathrm{s}}$ verifies $|s_i|\leq B_i$.
	
	Then, if we are given a point $\vec{\mathrm{x}}$ such that there exists $\vec{\mathrm{s}}$ with $\vec{\mathrm{v}}=\vec{\mathrm{A}}\vec{\mathrm{s}}-\vec{\mathrm{x}}$, $||\vec{\mathrm{v}}||\leq \sqrt{1/\pi} \alpha q/\sigma$, $|s_i|\leq B_i$ and $\rho_{\sigma/q}(\mathrm{\Lambda}\setminus\{\vec{0}\}+\vec{\mathrm{v}})\leq \xi\exp(-\alpha^2)/2$, we are able to find $\vec{\mathrm{s}}$ in at most $mn$ calls to the Gaussian sampling oracle, 
$n$ calls to the $\mathsf{LWE}$ solving oracle, with a probability of failure $n\sqrt{m}\epsilon+2^{-\mathrm{\Omega}(n)}$ and complexity $\bigo(mn^3+n^c)$ for some $c$.
\end{lemma}
\begin{proof}
	Let $\vec{\mathrm{y}}$ be sampled according to $D_{\mathrm{\Lambda}^*,\sigma}$, $\vec{\mathrm{v}}=\vec{\mathrm{As}}-\vec{\mathrm{x}}$.
	Then \[ \langle \vec{\mathrm{y}},\vec{\mathrm{x}}\rangle =\langle \vec{\mathrm{y}},\vec{\mathrm{As}}\rangle +\langle \vec{\mathrm{y}},\vec{\mathrm{v}}\rangle =\langle \vec{\mathrm{A}}^t\vec{\mathrm{y}},\vec{\mathrm{s}}\rangle -\langle \vec{\mathrm{y}},\vec{\mathrm{v}}\rangle \] and since $\vec{\mathrm{y}} \in \mathrm{\Lambda}^*$, $\vec{\mathrm{A}}^t\vec{\mathrm{y}} \in \Z^n$.
	We can thus provide the sample $(\vec{\mathrm{A}}^t\vec{\mathrm{y}},\langle \vec{\mathrm{y}},\vec{\mathrm{x}}\rangle )$ to the $\mathsf{LWE}$ solving oracle.

	The probability of obtaining the sample $\vec{\mathrm{A}}^t\vec{\mathrm{y}} \in (\Zq)^n$ is proportional to $\rho_\sigma(q\mathrm{\Lambda}^*+\vec{\mathrm{y}})$.
	Using the Poisson summation formula, 
	\[ \rho_\sigma(q\mathrm{\Lambda}^*+\vec{\mathrm{y}})=s^n \det(\mathrm{\Lambda}/q) \sum_{\vec{\mathrm{z}} \in \mathrm{\Lambda}/q} \exp(2i\pi\langle \vec{\mathrm{z}},\vec{\mathrm{y}}\rangle )\rho_{1/\sigma}(\vec{\mathrm{z}}). \]
	Since $\sigma \geq q\eta_\epsilon(\mathrm{\Lambda}^*)=\eta_\epsilon(q\mathrm{\Lambda}^*)$, we have
	\[ 1-\epsilon \leq \sum_{\vec{\mathrm{z}} \in \mathrm{\Lambda}/q} \cos(2\pi\langle \vec{\mathrm{z}},\vec{\mathrm{y}}\rangle )\rho_{1/\sigma}(\vec{\mathrm{z}}) \leq 1+\epsilon. \]
	Therefore, with \autoref{KLmajore}, the \textsf{KL}-divergence between the distribution of $\vec{\mathrm{A}}^t\vec{\mathrm{y}} \Mod q$ and the uniform distribution is less than $2\epsilon^2$.

	Also, with $\vec{\mathrm{c}} \in \vec{\mathrm{y}}+q\mathrm{\Lambda}^*$, \[ \E[\exp(2i\pi(\langle \vec{\mathrm{A}}^t\vec{\mathrm{y}},\vec{\mathrm{s}}\rangle -\langle \vec{\mathrm{y}},\vec{\mathrm{x}}\rangle )/q)|\vec{\mathrm{A}}^t\vec{\mathrm{y}} \Mod q]=\E_{D_{q\mathrm{\Lambda}*+\vec{\mathrm{c}},\sigma}}[\exp(2i\pi\langle \vec{\mathrm{y}},\vec{\mathrm{v}}/q\rangle )]. \]
	Let $f(\vec{\mathrm{y}})=\exp(2i\pi\langle \vec{\mathrm{y}},\vec{\mathrm{v}}/q\rangle )\rho_\sigma(\vec{\mathrm{y}})$ so that the bias is $f(q\mathrm{\Lambda}^*+\vec{\mathrm{c}})/\rho_\sigma(q\mathrm{\Lambda}^*+\vec{\mathrm{c}}).$
	Using the Poisson summation formula on both terms, this is equal to
	\[ \bigg(\sum_{\vec{\mathrm{y}} \in \mathrm{\Lambda}/q} \rho_{1/\sigma}(\vec{\mathrm{y}}-\vec{\mathrm{v}}/q)\exp(-2i\pi\langle \vec{\mathrm{y}},\vec{\mathrm{c}}\rangle )\bigg) \bigg/\bigg(\sum_{\vec{\mathrm{y}} \in \mathrm{\Lambda}/q} \rho_{1/\sigma}(\vec{\mathrm{y}})\cos(2\pi\langle \vec{\mathrm{y}},\vec{\mathrm{c}}\rangle \bigg). \]
	In this fraction, the numerator is at distance at most $\xi\exp(-\alpha^2)/(1+\epsilon)$ to $\exp(-\pi\sigma^2/q^2||\vec{\mathrm{v}}||^2)\geq\exp(-\alpha^2)$ and the denominator is in $[1-\epsilon;1+\epsilon]$.

	Using \autoref{KLerreur} with an algorithm which tests if the returned secret is equal to the real one, the failure probability is bounded by $\sqrt{m}\epsilon$.
	Thus, the $\mathsf{LWE}$ solving oracle works and gives $\vec{\mathrm{s}} \Mod q$.

\noindent
	Let $\vec{\mathrm{x'}}=(\vec{\mathrm{x}}-\vec{\mathrm{A}}(\vec{\mathrm{s}} \Mod q))/q$ so that $\vec{\mathrm{s'}}=(\vec{\mathrm{s}}-(\vec{\mathrm{s}} \Mod q))/q$ and $||\vec{\mathrm{v'}}||\leq ||\vec{\mathrm{v}}||/q$.
	Therefore, the reduction also works.
	If we repeat this process $n$ times, we can solve the last closest vector problem with Babai's algorithm, which reveals $\vec{\mathrm{s}}$.
\end{proof}

In the previous lemma, we required access to a $D_{\mathrm{\Lambda}^*,\sigma}$ oracle. However, for large enough $\sigma$, this hypothesis comes for free, as shown by the following lemma, which we borrow from \cite{brakerski2013classical}.

\begin{lemma}\label{lemma:freesampling}
	If we have a basis $\vec{\mathrm{A}}$ of the lattice $\mathrm{\Lambda}$, then for $\sigma\geq \bigo(\sqrt{\log n}||\widetilde{\vec{\mathrm{A}}}||)$, it is possible to sample in polynomial time from $D_{\mathrm{\Lambda},\sigma}$.
\end{lemma}

We will also need the following lemma, due to Banaszczyk~\cite{ban93}.
\begin{lemma}\label{lemma:ban93}
For a lattice $\mathrm{\Lambda}$, $\vec{\mathrm{c}}\in \R^n$, and $t\geq 1$,
\[ \frac{\rho \big((\mathrm{\Lambda}+\vec{\mathrm{c}})\setminus \mathcal{B}\big(0,t\sqrt{\frac{n}{2\pi}}\big)\big)}{\rho(\mathrm{\Lambda})} \leq \exp\big(-n(t^2-2\ln t-1)/2\big) \leq \exp\big(-n(t-1)^2/2\big).\]
\end{lemma}
\begin{proof}
	For any $s \geq 1$, using Poisson summation :
	\begin{align*}
		\frac{\rho_{s}(\mathrm{\Lambda}+\vec{\mathrm{c}})}{\rho(\mathrm{\Lambda})}= & s^n \frac{\sum_{\vec{\mathrm{x}}\in \mathrm{\Lambda}^*} \rho_{1/s}(\vec{\mathrm{x}}) \exp(2i\pi \langle \vec{\mathrm{x}},\vec{\mathrm{c}} \rangle)}{\rho(\mathrm{\Lambda}^*)} \\
							& \leq s^n \frac{\sum_{\vec{\mathrm{x}}\in \mathrm{\Lambda}^*} \rho_{1/s}(\vec{\mathrm{x}})}{\sum_{\vec{\mathrm{x}} \in \mathrm{\Lambda}^*} \rho(\vec{\mathrm{x}})} \\
		     & \leq s^n
	\end{align*}

	Then,
	\begin{align*}
		s^n \rho(\mathrm{\Lambda}) \geq & \rho_s\bigg((\mathrm{\Lambda}+\vec{\mathrm{c}})\setminus \mathcal{B}\bigg(\vec{0},t\sqrt{\frac{n}{2\pi}}\bigg)\bigg) \\
		\geq & \exp(t^2(1-1/s^2)n/2)\rho\bigg((\mathrm{\Lambda}+\vec{\mathrm{c}})\setminus \mathcal{B}\bigg(\vec{0},t\sqrt{\frac{n}{2\pi}}\bigg)\bigg).
	\end{align*}
	And therefore, using $s=t$ :
	\begin{align*}
		\frac{\rho((\mathrm{\Lambda}+\vec{\mathrm{c}})\setminus \mathcal{B}(\vec{0},t\sqrt{\frac{n}{2\pi}}))}{\rho(\mathrm{\Lambda})} \leq & \exp(-n(t^2(1-1/s^2)-2\ln s)/2) \\
						= & \exp(-n(t^2-2\ln t-1)/2) \\
		\leq & \exp(-n(t-1)^2/2),
	\end{align*}
where the last inequality stems from $\ln t\leq t-1$.
\end{proof}

\begin{theorem}
	Assume we have a $\mathsf{LWE}$ solving oracle of modulus $q\geq 2^n$, parameters $\beta$ and $\xi$ which needs $m$ samples.

	If we have a basis $\vec{\mathrm{A}}$ of the lattice $\mathrm{\Lambda}$, and a point $\vec{\mathrm{x}}$ such that $\vec{\mathrm{A}}\vec{\mathrm{s}}-\vec{\mathrm{x}}=\vec{\mathrm{v}}$ with $||\vec{\mathrm{v}}||\leq (1-1/n)\lambda_1/\beta/t < \lambda_1/2$ and $4\exp(-n(t-1/\beta-1)^2/2)\leq \xi\exp(-n/2/\beta^2)$, then with $n^2$ calls to the $\mathsf{LWE}$ solving oracle with secret $\vec{\mathrm{s}}$, we can find $\vec{\mathrm{s}}$ with probability of failure $2\sqrt{m}\exp(-n(t^2-2\ln t-1)/2)$ for any $t\geq 1+1/\beta$.
\end{theorem}
\begin{proof}
	Using Lemma~\ref{lemma:ban93}, we can prove that $\sigma=t\sqrt{n/2/\pi}/\lambda_1 \leq \eta_{\epsilon}(\mathrm{\Lambda}^*)$ 
	for  $\epsilon=2\exp(-n(t^2-2\ln t-1)/2)$ and 
\[ \rho_{1/\sigma}\big(\mathrm{\Lambda}\setminus\{\vec{0}\}+\vec{\mathrm{v}}\big) \leq 2\exp\big(-n(t(1-1/\beta/t)-1)^2/2\big). \]

	Using \textsf{LLL}, we can find a basis $\vec{\mathrm{B}}$ of $\mathrm{\Lambda}$ so that $||\widetilde{\vec{\mathrm{B}}^*}||\leq 2^{n/2}/\lambda_1$, and therefore, it is possible to sample in polynomial time from $D_{\mathrm{\Lambda},q\sigma}$ since $q \geq 2^n$ for sufficiently large $n$.

	The \textsf{LLL} algorithm also gives a non zero lattice vector of norm $\ell \leq 2^n\lambda_1$.
	For $i$ from $0$ to $n^2$, we let $\lambda=\ell(1-1/n)^i$, we use the algorithm of \autoref{LemmeBDDreduction} with standard deviation $tq\sqrt{n/2/\pi}/\lambda$, which uses only one call to the \textsf{LWE} solving oracle, and return the closest lattice vector of $\vec{\mathrm{x}}$ in all calls.

	Since $\ell(1-1/n)^{n^2}\leq 2^n\exp(-n)\lambda_1 \leq \lambda_1$, with $0\leq i\leq n^2$ be the smallest integer such that $\lambda=\ell(1-1/n)^i \leq \lambda_1$, we have $\lambda\geq (1-1/n)\lambda_1$.
	Then the lemma applies since \[ ||\vec{\mathrm{v}}|| \leq (1-1/n) \lambda_1/\beta/t \leq \sqrt{1/\pi} \sqrt{n/2}/\beta q/(tq\sqrt{n/2/\pi}/\lambda)=\lambda/t/\beta .\]
	Finally, the distance bound makes $\vec{\mathrm{As}}$ the unique closest lattice point of $\vec{\mathrm{x}}$.
\end{proof}
Using self-reduction, it is possible to remove the $1-1/n$ factor \cite{Lyubashevsky2009bounded}.

\begin{corollary}\label{BDD}
	It is possible to solve $\BDD_{B,\beta}^{||.||_\infty}$ in time $2^{(n/2+o(n))/\ln(1+\log \beta/\log B)}$ if $\beta=\omega(1)$, $\beta=2^{o(n/\log n)}$ and $\log B=\bigo(\log \beta)$.
\end{corollary}
\begin{proof}
	Apply the previous theorem and \autoref{thm:lwe} with some sufficiently large constant for $t$, and remark that dividing $\beta$ by some constant does not change the complexity.
\end{proof}

Note that since we can solve $\mathsf{LWE}$ for many secrets in essentially the same time than for one, we have the same property for $\BDD$.

\subsection{$\UniqueSVP$}

The following reduction was given in \cite{Lyubashevsky2009bounded} and works without modification.

\begin{theorem}
Given a $\BDD_{B,\beta}^{||.||_\infty}$ oracle, it is possible to solve $\UniqueSVP_{B,\beta}^{||.||_\infty}$ in polynomial time of $n$ and $\beta$.
\end{theorem}
\begin{proof}
	Let $2\beta\geq p\geq \beta$ be a prime number, $\vec{\mathrm{A}}$ and $\vec{\mathrm{s}}$ such that $||\vec{\mathrm{A}}\vec{\mathrm{s}}||=\lambda_1$.
	If $\vec{\mathrm{s}}=\vec{\mathrm{0}} \mod p$ then $\vec{\mathrm{A}}(\vec{\mathrm{s}}/p)$ is a non zero lattice vector, shorter 
	than $||\vec{\mathrm{A}}\vec{\mathrm{s}}||$, which is impossible.
	Let $i$ such that $s_i\neq 0 \mod p$, and \[ \vec{\mathrm{A'}}_i=[ \vec{\mathrm{a}}_0,\dots,\vec{\mathrm{a}}_{i-1},p\vec{\mathrm{a}}_i, \vec{\mathrm{a}}_{i+1},\ldots ,\vec{\mathrm{a}}_{n-1}], \] which generates a sublattice $\mathrm{\Lambda'}$ of $\mathrm{\Lambda}$.
	If $k\neq 0 \mod p$, then $k\vec{\mathrm{A}}\vec{\mathrm{s}} \not\in \mathrm{\Lambda'}$ so that $\lambda_1(\mathrm{\Lambda'})\geq \lambda_2(\mathrm{\Lambda})$.

	If $\vec{\mathrm{s'}}_i=(s_0,\ldots,\lfloor s_i/p\rceil,\ldots,s_{n-1})$, then $\vec{\mathrm{A'}}_i\vec{\mathrm{s'}}_i+(s_i \Mod p)\vec{\mathrm a}_i=\vec{\mathrm{A}}\vec{\mathrm{s}}$ and $||\vec{\mathrm{s'}}_i||_{\infty} \leq B$.
	Therefore, calling $\BDD$ to find the point of $\mathrm{\Lambda'}$ closest to $(s_{i} \mod p)\vec{\mathrm a}_{i}$ yields $\vec{\mathrm s'}$. By trying every value of $i$ and $(s_{i}\mod p)$, $np$ calls to $\BDD_{B,\beta}^{||.||_\infty}$ are enough to solve $\UniqueSVP_{B,\beta}^{||.||_\infty}$.
\end{proof}

The reductions for both $\BDD^{||.||}$ and $\UniqueSVP^{||.||}$ work the same way.

\subsection{$\GapSVP$}
The following reduction is a modification of the reduction given in \cite{peikert2009public}, which comes from \cite{goldreich1998limits}, but doesn't fit to our context.

First, we need a lemma proven in \cite{goldreich1998limits} :
\begin{lemma}
	The volume of the intersection of two balls of radius $1$ divided by the volume of one ball is at least
	\[ \frac{d}{3}(1-d^2)^{(n-1)/2}\sqrt{n} \geq \frac{d}{3}(1-d^2)^{n/2} \]
	where $d\leq 1$ is the distance between the centers.
\end{lemma}

\begin{lemma}
	\label{LemmeGapReduction}
	Given access to an oracle solving $\BDD_{B,\beta}^{||.||_\infty}$ and let $\Dis$ be an efficiently samplable distribution over 
	points of $\Z^n$ whose infinity norm is smaller than $b$ and for $\xi \leq 1$ and let
	\[ \epsilon=\min_{||\vec{\mathrm{s}}||_{\infty} \leq R} \Pr_{\vec{\mathrm{x}}\sim \Dis} \bigg[\xi \leq \frac{\Dis(\vec{\mathrm{x}}+\vec{\mathrm{s}})}{\Dis(\vec{\mathrm{x}})} \leq 1/\xi\bigg]. \]

	Then, we can solve $\GapSVP_{R,\beta/d}^{||.||_\infty}$ with negligible probability of failure using $\bigo(dn/\epsilon/\xi(1-d^2)^{-n/2})$ calls to the $\BDD$ oracle.
\end{lemma}
\begin{proof}
	Let $K=\Theta(dn/\epsilon/\xi(1-d^2)^{-n/2})$, $\vec{\mathrm{A}}$ a basis of the lattice $\mathrm{\Lambda}$.
	The algorithm consists in testing $K$ times the oracle : sample $\vec{\mathrm{x}}$ with law $\Dis$ and check if $\BDD(\vec{\mathrm{A}}\vec{\mathrm{x}}+\vec{\mathrm{e}})=\vec{\mathrm{x}}$ for $\vec{\mathrm{e}}$ sampled uniformly within a ball centered at the origin and of radius $1/d$.
	If any check is wrong, return that $\lambda_1 \leq 1$, else return $\lambda_1 \geq \beta/d$.

	Clearly, the algorithm runs in the given complexity, and if $\lambda_1 \geq \beta/d$, it is correct.

	So, let $\vec{\mathrm{s}}\neq \vec{0}$ such that $||\vec{\mathrm{A}}\vec{\mathrm{s}}||\leq 1$.
	Let $\vec{\mathrm{x}}$ and $\vec{\mathrm{e}}$ be sampled as in the algorithm.
	With probability over $\vec{\mathrm{e}}$ greater than $(1-d^2)^{n/2}d/3$, $||\vec{\mathrm{A}}\vec{\mathrm{x}}+\vec{\mathrm{e}}-\vec{\mathrm{A}}(\vec{\mathrm{x}}+\vec{\mathrm{s}})||\leq 1/d$.
	We condition on this event.

	Then, with probability at least $\epsilon$, $1/\xi \geq k=\frac{D(\vec{\mathrm{x}}+\vec{\mathrm{s}})}{D(\vec{\mathrm{x}})} \geq \xi$ and we also condition on this event.
	Let $p$ be the probability that $\BDD(\vec{\mathrm{Ax}}+\vec{\mathrm{e}})=\vec{\mathrm{x}}$.
	The probability of failure is at least
	\[ (1-p)/(1+k)+pk/(1+k) \geq \min(k,1)/(1+k) \geq \xi. \]
	Therefore, the probability of failure of the algorithm is $2^{-\mathrm{\Omega}(n)}$.
\end{proof}

We could have used the uniform distribution for $\Dis$, but a discrete gaussian is more efficient.

\begin{lemma}
	Let $\mathrm{\Lambda}$ be a one dimensional lattice.
	Then, \[ \E_{\vec{\mathrm{x}}\sim D_{\mathrm{\Lambda}}}[||\vec{\mathrm{x}}||^2]\leq \frac{1}{2\pi}. \]
\end{lemma}
\begin{proof}
	Let $f(\vec{\mathrm{x}})=||\vec{\mathrm{x}}||^2\rho(\vec{\mathrm{x}})$.
	Then, using a Poisson summation :
	\[ \frac{f(\mathrm{\Lambda})}{\rho(\mathrm{\Lambda})}=\frac{\sum_{\vec{\mathrm{x}}\in \mathrm{\Lambda}^*} \big(\frac{1}{2\pi}-||\vec{\mathrm{x}}||^2\big)\rho(\vec{\mathrm{x}})}{\rho(\mathrm{\Lambda}^*)}\leq \frac{1}{2\pi}. \]
\end{proof}

\begin{lemma}
	Let $\Dis$ be such that $\Dis(\vec{\mathrm{x}}) \propto \exp(-||\vec{\mathrm{x}}||^2/(2\sigma^2))$ for all $\vec{\mathrm{x}}$ with $||\vec{\mathrm{x}}||_{\infty}\leq B$.
	Then $\Dis$ is polynomially samplable for $B-R\geq 2\sigma$, and using the definitions of \autoref{LemmeGapReduction}, we have for $\xi=\exp(-nR^2/(2\sigma^2)-2\sqrt{n}R/\sigma)$, $\epsilon \geq \exp(-2n\exp(-((B-R)/\sigma-1)^2/2))/2$.
\end{lemma}
\begin{proof}
	Using the Banaszczyk lemma , we have :
	\[ \frac{\sum_{x=R-B}^{B-R}\exp(-x^2/(2\sigma^2))}{\sum_{x=-B}^B \exp(-x^2/(2\sigma^2))} \geq \frac{\sum_{x=R-B}^{B-R}\exp(-x^2/(2\sigma^2))}{\sum_{x \in \Z}\exp(-x^2/(2\sigma^2))} \geq 1-\exp(-((B-R)/\sigma-1)^2/2). \]
	So that :
	\[ \Pr_{\vec{\mathrm{x}}\sim \Dis}[||\vec{\mathrm{x}}||_{\infty}\leq B-R] \geq \exp(-2n\exp(-((B-R)/\sigma-1)^2/2)). \]
	Since the discrete Gaussian distribution over $\Z$ is polynomially samplable, this is also the case for $\Dis$ since $B-R \geq 2\sigma$.

	We now condition the distribution over $||\vec{\mathrm{x}}||_{\infty}\leq B-R$.
	For some $\vec{\mathrm{s}}$ such that $||\vec{\mathrm{s}}||_{\infty}\leq R$ and $N=||\vec{\mathrm{s}}||$, we study the variable \[ \ell=2\sigma^2\ln(\exp(-||\vec{\mathrm{x}}+\vec{\mathrm{s}}||^2/(2\sigma^2))/\exp(-||\vec{\mathrm{x}}||^2/(2\sigma^2)))=N^2-2\langle \vec{\mathrm{x}},\vec{\mathrm{s}}\rangle. \]
	By symmetry :
	\[ \E[\ell]=\E[N^2-2\langle x,s\rangle ]=N^2. \]
	Let $\Dis'$ be the distribution over the first coordinate and $v=s_0$.
	Then, using the previous lemma :
	\[ \Var[v^2-xv]=\E_{x\sim \Dis'}[(v^2-xv-v^2)^2]\leq \E_{x\sim D_{\Z,\sigma/\sqrt{2\pi}}}[x^2]v^2 \leq \sigma^2v^2. \]
	Summing over all coordinates, we have :
	\[ \Var[\ell]\leq 4\sigma^2N^2. \]
	By the Chebyshev inequality, we get :
	\[ \Pr[|\ell+N^2|\geq 4N\sigma]\leq \frac{1}{2}. \]
	And the claim follows.
\end{proof}

\begin{theorem}
	One can solve any $\GapSVP_{o(B\sqrt{\log \log \log \beta/\log \log \beta}),\beta}^{||.||_\infty}$ in time \[ 2^{(n/2+o(n))/\ln(1+\log \beta/\log B)} \] for $\beta=2^{o(n/\log n)}$, $\beta=\omega(1)$, $B\geq 2$.
\end{theorem}
\begin{proof}
	Use the previous lemma with $\sigma=B/\sqrt{3\ln \log \log \beta}$ and~\autoref{BDD} with $\beta'=\beta/\log(\beta)$, so that it is sufficient to decode $2^{o(n/\log \log \beta)}$ points.
\end{proof}
\begin{theorem}
	If it is possible to solve $\BDD_{B,\beta}^{||.||_\infty}$ in polynomial time, then it is possible to solve in randomized polynomial time $\GapSVP_{B/\sqrt{n},\beta\sqrt{n/\log n}}^{||.||_\infty}$.
\end{theorem}
\begin{proof}
	Use $\sigma=B/\sqrt{3\ln n}$.
\end{proof}

We now proves the corresponding theorem in norm L2.

\begin{theorem}
	Let $\Dis$ be the uniform distribution over the points $\vec{x}$ of $\Z^n$ with $||\vec{x}||\leq B$.
	Then, $\Dis$ is samplable in time $\bigo(n^2B^2\log B)$ and for any $\vec{s}$ with $||\vec{s}||\leq R$ and $\vec{x}$ sampled according to $\Dis$, we have for $n$ sufficiently large
	\[ \Pr[||\vec{x}+\vec{s}||\leq B]\geq (\sqrt{B^2-R^2}-\sqrt{n}/2)^{n-1}(R-\sqrt{n}/2)/(B+\sqrt{n}/2)^n. \]
\end{theorem}
\begin{proof}
	It is easy to see that a dynamic programming algorithm using $\bigo(nB^2)$ operations on integers smaller than $B^n$ can determine the $i$-th lattice point of the ball, ordered lexicographically.
	Therefore, $\Dis$ is samplable.

	Let $E$ be the right circular cylinder of center $\vec{s}/2$, with generating segments parallel to $\vec{s}$, of length $R$ and radius $r=\sqrt{B^2-R^2}$, and $H$ the same cylinder with length $R-\sqrt{n}/2$ and radius $r-\sqrt{n}/2$.
	For the lattice point $\vec{x} \in \Z^n$, let $F_{\vec{x}}$ be the axis-aligned cube of length $1$ centered on $\vec{x}$.
	Let $F$ be the union of all $F_{\vec{x}}$ such that $F_{\vec{x}} \subset E$.
	We have $|E\cap \Z^n|=\vol(F)$ and $H \subset F$.
	Therefore, $|E\cap \Z^n|\geq V_{n-1}(r-\sqrt{n}/2)^{n-1}(R-\sqrt{n}/2)$.
	Also, $E$ is a subset of the intersection of the balls of radius $B$ and centers $\vec{0}$ and $\vec{s}$.
	Using~\autoref{lemme:cardBoule} and $V_{n-1}\geq V_n$ for $n$ sufficiently large, the result follows.  
\end{proof}

\begin{corollary}
	One can solve $\GapSVP_{\sqrt{n}o(B/\sqrt{\log \log \beta}),\beta}^{||.||}$, $\beta=\omega(1)$ and $\beta=2^{o(n/\log n)}$ in time
	\[ 2^{(n/2+o(n))/\ln(1+\log \beta/\log B)}. \]
\end{corollary}
\begin{proof}
	Apply~\autoref{thm:lweL2} with a reduction to $\BDD$ with parameter $\beta'=\beta/\log(\beta)$ and $B'=\max(B,\log \beta)$.
	Then, apply the previous theorem with~\autoref{LemmeGapReduction}. 
\end{proof}

\begin{corollary}
It is possible to solve any $\GapSVP_{\sqrt{n}2^{\sqrt{\log n}},n^c}^{||.||}$ with $c>0$ in time $2^{(n+o(n))/\ln \ln n}.$
\end{corollary}
\begin{proof}
Use the previous corollary with $B=2^{\sqrt{\log n}}\log \log n$ and $\beta=n^c$.
\end{proof}

\begin{theorem}
If it is possible to solve $\BDD_{B,\beta}^{||.||_\infty}$ in polynomial time, then it is possible to solve in randomized polynomial time $\GapSVP_{B/\sqrt{n},\beta\sqrt{n/\log n}}^{||.||_\infty}$.
\end{theorem}

\section{Other applications}

\subsection{Low density subset-sum problem}
\label{sec:subsetsum}

\begin{definition}
	We are given a vector $\vec{\mathrm{a}}\in \Z^n$ whose coordinates are sampled independently and uniformly in $[0;M)$, and $\langle \vec{\mathrm{a}},\vec{\mathrm{s}} \rangle$ where the coordinates of $\vec{\mathrm{s}}$ are sampled independently and uniformly in $\zeroun$.
	The goal is to find $\vec{\mathrm{s}}$.
	The density is defined as $d=\frac{n}{\log M}$.
\end{definition}
Note that this problem is trivially equivalent to the \textit{modular} subset-sum problem, where we are given $\langle \vec{a},\vec{s} \rangle \Mod M$ by trying all possible $\lfloor \langle \vec{a},\vec{s} \rangle/M \rfloor$ .

In~\cite{lagarias1985solving,coster1992improved}, Lagarias \textit{et al.} reduce the subset sum problem to 
$\mathsf{UniqueSVP}$, even though this problem was not defined at that time. We will show a reduction to $\BDD_{1,\mathrm{\mathrm{\Omega}}(2^d)}^{||.||_\infty}$, which is essentially the same.
First, we need two geometric lemmata.
\begin{lemma} \label{lemme:cardBoule}
	Let $\mathcal{B}_n(r)$, the number of points of $\Z^n$ of norm smaller than $r$, and $V_n$ the volume of the unit ball.
	Then,
	\[ \mathcal{B}_n(r)\leq V_n\bigg( r+\frac{\sqrt{n}}{2} \bigg)^n. \]
\end{lemma}
\begin{proof}
	For each $\vec{\mathrm{x}}\in \Z^n$, let $E_{\vec{\mathrm{x}}}$ be a cube of length $1$ centered on $\vec{\mathrm{x}}$.
	Let $E$ be the union of all the $E_{\vec{\mathrm{x}}}$ which have a non empty intersection with the ball of center 
	$\vec{0}$ and radius $r$. 
	Therefore $\vol(E)\geq \mathcal{B}_n(r)$ and since $E$ is included in the ball of center $\vec{0}$ and radius $r+\frac{\sqrt{n}}{2}$, the claim is proven.
\end{proof}
\begin{lemma}\label{lemme:volBoule}
	For $n\geq 4$ we have \[ V_n= \frac{\pi^{n/2}}{(n/2)!} \leq (\sqrt{\pi\me/n})^n \]
\end{lemma}

\begin{theorem}
	Using one call to a $\BDD_{1,c2^{1/d}}^{||.||_\infty}$ oracle with any $c< \sqrt{2/\pi/\me}$ and $d=o(1)$, and polynomial time, it is possible to solve a subset-sum problem of density $d$, with negligible probability of failure.
\end{theorem}
\begin{proof}
	With the matrix :
	\[ \vec{\mathrm{A}}=\begin{pmatrix} \vec{\mathrm{I}} \\ C\vec{\mathrm{a}} \end{pmatrix} \]
	for some $C>c2^{1/d}\sqrt{n}/2$ and $\vec{\mathrm{b}}=(1/2,\dots,1/2,C\langle \vec{\mathrm{a}},\vec{\mathrm{s}}\rangle )$, return $\BDD(\vec{\mathrm{A}},\vec{\mathrm{b}})$.
	It is clear that $||\vec{\mathrm{A}}\vec{\mathrm{s}}-\vec{\mathrm{b}}||= \sqrt{n}/2$.
	Now, let $\vec{\mathrm{x}}$ such that $||\vec{\mathrm{A}}\vec{\mathrm{x}}||=\lambda_1$.
	If $\langle \vec{\mathrm{a}},\vec{\mathrm{x}}\rangle \neq 0$, then $\lambda_1=||\vec{\mathrm{A}}\vec{\mathrm{x}}||\geq C$ therefore $\beta\geq c2^{1/d}$.
	Else, $\langle \vec{\mathrm{a}},\vec{\mathrm{x}}\rangle=0$.
	Without loss of generality, $x_0\neq 0$, we let $y=-(\sum_{i>0} a_ix_i)/x_0$ and the probability over $\vec{\mathrm{a}}$ that $\langle \vec{\mathrm{a}},\vec{\mathrm{x}}\rangle =0$ is :
	\[ \Pr[\langle \vec{\mathrm{a}},\vec{\mathrm{x}}\rangle =0]=\Pr[a_0=y]=\sum_{z=0}^{M-1} \Pr[y=z]\Pr[a_0=z] \leq \frac{1}{M}. \]
	Therefore, the probability of failure is at most, for sufficiently large $n$,
	\begin{align*}
		\mathcal{B}_n(\beta\sqrt{n}/2)/M \leq & (\sqrt{\pi\me/n})^n(c2^{1/d}\sqrt{n}/2+\sqrt{n}/2)^n/2^{n/d} \\
		= & \big(\sqrt{\pi\me/2}(c+2^{-1/d})\big)^n=2^{-\mathrm{\Omega}(n)}.\qedhere
	\end{align*}
\end{proof}

\begin{corollary}
	For any $d=o(1)$ and $d=\omega(\log n/n)$, we can solve the subset-sum problem of density $d$ with negligible probability of failure in time $2^{(n/2+o(n))/\ln(1/d)}.$
\end{corollary}

The cryptosystem of Lyubashevsky \textit{et al.}~\cite{lyubashevsky2010public} uses $2^{1/d}>10n\log^2 n$ and is therefore broken in time $2^{(\ln 2/2+o(1))n/\log \log n}$.
Current lattice reduction algorithms are slower than this one when $d=\omega(1/(\log n\log \log n))$.

\subsection{Sample Expander and application to $\mathsf{LWE}$ with binary errors}

\begin{definition}
	Let $q$ be a prime number.
	The problem $\mathsf{Small\mhyphen DecisionLWE}$ is to distinguish $(\vec{\mathrm{A}},\vec{\mathrm{b}})$ with $\vec{\mathrm{A}}$ sampled uniformly with $n$ columns and $m$ rows, $\vec{\mathrm{b}}=\vec{\mathrm{As}}+\vec{\mathrm{e}}$ such that $||\vec{\mathrm{s}}||^2+||\vec{\mathrm{e}}||^2\leq nk^2$ and $||\vec{\mathrm{s}}||\leq \sqrt{n}B$ from $(\vec{\mathrm{A}},\vec{\mathrm{b}})$ sampled uniformly.
	Also, the distribution $(\vec{\mathrm{s}},\vec{\mathrm{e}})$ is efficiently samplable.

	The problem $\mathsf{Small\mhyphen SearchLWE}$ is to find $\vec{\mathrm{s}}$ given $(\vec{\mathrm{A}},\vec{\mathrm{b}})$ with $\vec{\mathrm{A}}$ sampled uniformly and $\vec{\mathrm{b}}=\vec{\mathrm{As}}+\vec{\mathrm{e}}$ with the same conditions on $\vec{\mathrm{s}}$ and $\vec{\mathrm{e}}$.
\end{definition}

These problems are generalizations of $\mathsf{BinaryLWE}$ where $\vec{\mathrm{s}}$ and $\vec{\mathrm{e}}$ have coordinates sampled uniformly in $\zeroun$.
In this case, remark that each sample is a root of a known quadratic polynomial in the coordinates of $\vec{\mathrm{s}}$.
Therefore, it is easy to solve this problem when $m\geq n^2$.
For $m=\bigo(n)$, a Gr\"{o}bner basis algorithm applied on this system will (heuristically) have a complexity of $2^{\mathrm{\Omega}(n)}$~\cite{albrechtalgebraic}.
For $m=\bigo(n/\log n)$ and $q=n^{\bigo(1)}$, it has been shown to be harder than a lattice problem in dimension $\Theta(n/\log n)$~\cite{micciancio2013hardness}.

We will first prove the following 
theorem~\footnote{In~\cite{duc2015better}, the authors gave a short justification of a similar claim which is far from proven.}, 
with the coordinates of $\vec{\mathrm{x}}$ and $\vec{\mathrm{y}}$ distributed according to a samplable $\Dis$ :
\begin{theorem}
	Assume there is an efficient distinguisher which uses $k$ samples for $\DecisionLWE$ (respectively a solver for $\SearchLWE$) with error distribution $\langle \vec{\mathrm{s}},\vec{\mathrm{y}}\rangle+\langle \vec{\mathrm{e}},\vec{\mathrm{x}}\rangle$ of advantage (resp. success probability) $\epsilon$.

	Then, either there is an efficient distinguisher for $\DecisionLWE$ with samples and secret taken uniformly, and error distribution $\Dis$ in dimension $m-1$ and with $n+m$ samples of advantage $\frac{\xi}{4qk}-q^{-n}-q^{-m}$;
	or there is an efficient distinguisher of advantage $\epsilon-\xi$ for $\mathsf{Small\mhyphen Decision\mhyphen LWE}$ (resp. solver of success probability $\epsilon-\xi$ for $\mathsf{Small\mhyphen Search\mhyphen LWE}$).
\end{theorem}

\subsubsection{Sample Expander}
\label{sec:expansion}
The reduction is a generalization and impovement of D{\"o}ttling's reduction \cite{dottling2015low}, which was made for 
$\mathsf{LPN}$.

\begin{lemma}[Hybrid lemma~\cite{goldreich2004foundations}]
	We are given a distinguisher between the distributions $\Dis_0$ and $\Dis_1$ of advantage $\epsilon$ which needs $m$ samples.
	Then, there exists a distinguisher between the distributions $\Dis_0$ and $\Dis_1$ of advantage $\epsilon/m$ which needs one sample, $m$ samples of distribution $\Dis_0$ and $m$ samples of distribution $\Dis_1$.
\end{lemma}
\begin{proof}
	The distinguisher samples an integer $i$ uniformly between $0$ and $m-1$.
	It then returns the output of the oracle, whose input is $i$ samples of $\Dis_0$, the given sample, and $m-i-1$ samples of $\Dis_1$.

	Let $p_i$ the probability that the oracle outputs $1$ if its input is given by $i$ samples of $\Dis_0$ followed by $m-i$ samples of $\Dis_1$.
	Then, if the sample comes from $\Dis_0$, the probability of outputting $1$ is $\frac{1}{m}\sum_{i=1}^{m} p_i$.
	Else, it is $\frac{1}{m}\sum_{i=0}^{m-1} p_i$.
	Therefore, the advantage of our algorithm is $\frac{1}{m}|p_m-p_0|=\frac{\epsilon}{m}$.
\end{proof}

In the following, we will let, for $\Dis$ some samplable distribution over $\Zq$, $\vec{\mathrm{x}}$ be sampled according to $\Dis^m$, $\vec{\mathrm{y}}$ be sampled according to $\Dis^n$, $\vec{\mathrm{z}}$ according to $(\vec{\mathrm{e}}|\vec{\mathrm{s}})$, $\vec{\mathrm{u}}$ be sampled uniformly in $(\Zq)^m$, $v$ be sampled uniformly in $\Zq$, $\vec{\mathrm{w}}$ in $(\Zq)^n$.

\begin{definition}
	The problem $\mathsf{Knapsack\mhyphen LWE}$ is to distinguish between $(\vec{\mathrm{G}},\vec{\mathrm{c}})$ where $\vec{\mathrm{G}}$ is always sampled uniformly in $(\Zq)^{m\times (m+n)}$, and $\vec{\mathrm{c}}$ is either uniform, or sampled according to $\vec{\mathrm{G}}\Dis^{m+n}$.

	The problem $\mathsf{Extended\mhyphen LWE}$ is to distinguish between $(\vec{\mathrm{A}},^t\vec{\mathrm{x}}\vec{\mathrm{A}}+\vec{\mathrm{y}},\vec{\mathrm{z}},\langle (\vec{\mathrm{x}}|\vec{\mathrm{y}}),\vec{\mathrm{z}} \rangle)$ and $(\vec{\mathrm{A}},\vec{\mathrm{w}},\vec{\mathrm{z}},\langle (\vec{\mathrm{x}}|\vec{\mathrm{y}}),\vec{\mathrm{z}} \rangle)$.

	The problem $\mathsf{First\mhyphen is\mhyphen errorless\mhyphen LWE}$ is to distinguish between $(\vec{\mathrm{A}},^t\vec{\mathrm{x}}\vec{\mathrm{A}}+\vec{\mathrm{y}},\vec{\mathrm{u}},\langle \vec{\mathrm{x}},\vec{\mathrm{u}}\rangle)$ and $(\vec{\mathrm{A}},\vec{\mathrm{w}},\vec{\mathrm{u}},\langle \vec{\mathrm{x}},\vec{\mathrm{u}} \rangle)$.
\end{definition}

The following lemma comes from \cite{micciancio2011pseudorandom}, we give a sketch of their proof.
\begin{lemma}
	The problems $\mathsf{Knapsack\mhyphen LWE}$ and $\DecisionLWE$ in dimension $m$ with $n+m$ samples where are equivalent: there are reductions in both ways which reduce the advantage by $q^{-n-1}$.
\end{lemma}
\begin{proof}
	Given the $\DecisionLWE$ problem $(\vec{\mathrm{A}},\vec{\mathrm{b}})$, we sample a uniform basis 
	$\vec{\mathrm{G}}$ of the left kernel of $\vec{\mathrm{A}}$ and outputs $(\vec{\mathrm{G}},\vec{\mathrm{Gb}})$.
	If $\vec{\mathrm{b}}=\vec{\mathrm{As}}+\vec{\mathrm{e}}$, we have therefore $\vec{\mathrm{Gb}}$ distributed as $\vec{\mathrm{Ge}}$.

	Given the $\mathsf{Knapsack\mhyphen LWE}$ problem $(\vec{\mathrm{G}},\vec{\mathrm{c}})$, we sample a uniform basis $\vec{\mathrm{A}}$ of the right kernel of $\vec{\mathrm{A}}$, and a uniform $\vec{\mathrm{b}}$ such that $\vec{\mathrm{Gb}}=\vec{\mathrm{c}}$ and outputs $(\vec{\mathrm{A}},\vec{\mathrm{b}})$.
	If $\vec{\mathrm{c}}=\vec{\mathrm{Ge}}$, then $\vec{\mathrm{b}}$ is distributed as $\vec{\mathrm{As}}+\vec{\mathrm{e}}$ where $\vec{\mathrm{s}}$ is uniform.

	Both reductions map the distributions to their counterparts, except when $\vec{\mathrm{A}}$ or $\vec{\mathrm{G}}$ are not full rank, which happens with probability at most $q^{-n-1}$, hence the result.
\end{proof}

The lemma is a slight modification of the lemma given in \cite{alperin2012circular}.
\begin{lemma}
	Given a distinguisher of advantage $\epsilon$ for $\mathsf{ExtendedLWE}$, there is a distinguisher of advantage $\frac{\epsilon(1-1/q)-2q^{-n}}{q}$ for $\DecisionLWE$ in dimension $m$, uniform secret, with $n+m$ uniform samples of error distribution $\Dis$.
\end{lemma}
\begin{proof}
	Using the previous lemma, we can assume that we want to solve a $\mathsf{KnapsackLWE}$ problem.
	Let $(\vec{\mathrm{G}},\vec{\mathrm{c}})$ be its input, with $\vec{\mathrm{c}}=\vec{\mathrm{Ge}}$.
	We start by sampling $\vec{\mathrm{z}}$, $\vec{\mathrm{t}}$ uniformly over $(\Zq)^m$ and $\vec{\mathrm{e'}}$ according to $\Dis^{n+m}$.

	We then compute $\vec{\mathrm{G'}}=\vec{\mathrm{G}}-\vec{\mathrm{t}}^t\vec{\mathrm{z}}$ and $\vec{\mathrm{c'}}=\vec{\mathrm{c}}-\langle \vec{\mathrm{z}},\vec{\mathrm{e'}} \rangle \vec{\mathrm{t}}$.
	Remark that $\vec{\mathrm{G'}}$ is uniform and $\vec{\mathrm{c'}}=\vec{\mathrm{Ge}}-\langle \vec{\mathrm{z}},\vec{\mathrm{e'}} \rangle \vec{\mathrm{t}}=\vec{\mathrm{G'e}}+\langle \vec{\mathrm{z}},\vec{\mathrm{e}}-\vec{\mathrm{e'}} \rangle \vec{\mathrm{t}}$.
	Therefore, if $\langle \vec{\mathrm{z}},\vec{\mathrm{e'}}\rangle=\langle \vec{\mathrm{z}},\vec{\mathrm{e}} \rangle$, which happens with probability at least $\frac{1}{q}$, $(\vec{\mathrm{G'}},\vec{\mathrm{c'}})$ comes from the same distribution as $(\vec{\mathrm{G}},\vec{\mathrm{c}})$.
	Also, $\langle \vec{\mathrm{z}},\vec{\mathrm{e}} \rangle$ is known.
	Else, since $\vec{\mathrm{t}}$ is uniform, $(\vec{\mathrm{G'}},\vec{\mathrm{c'}})$ is uniform.

	We then perform the previous reduction to a $\DecisionLWE$ problem, where $\langle \vec{\mathrm{z}},\vec{\mathrm{e}} \rangle$ is known.
	Finally, we use \autoref{th:petitsecret} to reduce to $\mathsf{Extended\mhyphen LWE}$.
	If it doesn't output $n$ samples, which happens with probability at most $1/q$, our distinguisher returns a uniform boolean.
	Else, remark that we have $(\vec{\mathrm{x}}|\vec{\mathrm{y}})=\vec{\mathrm{e}}$ so that $\langle (\vec{\mathrm{x}}|\vec{\mathrm{y}}),\vec{\mathrm{z}} \rangle$ is known and we use the given distinguisher.
\end{proof}

\begin{lemma}
	Assume there is a distinguisher of advantage $\epsilon$ for $\mathsf{First\mhyphen is\mhyphen errorless\mhyphen LWE}$.
	Then, there is a distinguisher of advantage at least $\epsilon(1-1/q)-q^{-m}$ and $\DecisionLWE$ in dimension $m-1$ with uniform secret and $n+m$ uniform samples of error distribution $\Dis$.
\end{lemma}
\begin{proof}
	We take a $\DecisionLWE$ problem in dimension $m-1$ with $n+m$ samples, and extend it of one coordinate, the corresponding secret coordinate being sampled uniformly.
	We then switch the error and the secret, as in \autoref{th:petitsecret}, and return a random boolean if the reduction do not output $n$ samples.

	We have the errorless sample $(\vec{\mathrm{e}}_{m-1},s_{m-1})$ with $\vec{\mathrm{e}}_{m-1}=(0,\ldots,0,1)$.
	If we feed this sample to our reduction, it outputs an errorless sample $(\vec{\mathrm{u}},\langle \vec{\mathrm{x}},\vec{\mathrm{u}}\rangle)$, and the vectorial part is distributed as a column of a uniform invertible matrix.
	Therefore, the statistical distance between the distribution of $\vec{\mathrm{u}}$ and the uniform distribution is at most $q^{-m}$.
\end{proof}

\begin{theorem}
	Assume there is an efficient distinguisher which uses $k$ samples for $\DecisionLWE$ (respectively a solver for $\SearchLWE$) with error distribution $\langle \vec{\mathrm{s}},\vec{\mathrm{y}}\rangle+\langle \vec{\mathrm{e}},\vec{\mathrm{x}}\rangle$ of advantage (resp. success probability) $\epsilon$.

	Then, either there is an efficient distinguisher for $\DecisionLWE$ with samples and secret taken uniformly, and error distribution $\Dis$ in dimension $m-1$ and with $n+m$ samples of advantage $\frac{\xi}{4qk}-q^{-n}-q^{-m}$;
	or there is an efficient distinguisher of advantage $\epsilon-\xi$ for $\mathsf{Small\mhyphen Decision\mhyphen LWE}$ (resp. solver of success probability $\epsilon-\xi$ for $\mathsf{Small\mhyphen Search\mhyphen LWE}$).
\end{theorem}
\begin{proof}
	Let
	\begin{itemize}
		\item $\Dis_0$ be the distribution given by $(\vec{\mathrm{w}},\langle\vec{\mathrm{w}},\vec{\mathrm{s}}\rangle+\langle \vec{\mathrm{s}},\vec{\mathrm{y}} \rangle+\langle \vec{\mathrm{e}},\vec{\mathrm{x}} \rangle)$
		\item $\Dis_1$ be the distribution given by $(^t\vec{\mathrm{xA}}+\vec{\mathrm{y}},\langle \vec{\mathrm{x}},\vec{\mathrm{b}} \rangle)$
		\item $\Dis_2$ be the distribution given by $(^t\vec{\mathrm{xA}}+\vec{\mathrm{y}},\langle \vec{\mathrm{x}},\vec{\mathrm{u}} \rangle)$
		\item $\Dis_3$ be the distribution given by $(\vec{\mathrm{w}},v)$.
	\end{itemize}
	Let $\epsilon_i$ be the advantage of $D$ between $\Dis_i$ and $\Dis_{i+1}$.

	Let $(\vec{\mathrm{A}},\vec{\mathrm{c}},(\vec{\mathrm{e}}|\vec{\mathrm{s}}),\langle (\vec{\mathrm{x}}|\vec{\mathrm{y}}),(\vec{\mathrm{e}}|\vec{\mathrm{s}})\rangle)$ be an instance of the $\mathsf{Extended\mhyphen LWE}$ problem.
	We compute $\vec{\mathrm{b}}=\vec{\mathrm{As}}-\vec{\mathrm{e}}$, and use the hybrid lemma with distributions $\Dis_0$ and $\Dis_1$, and sample $(\vec{\mathrm{c}},\langle \vec{\mathrm{c}},\vec{\mathrm{s}} \rangle+\langle (\vec{\mathrm{x}}|\vec{\mathrm{y}}),(\vec{\mathrm{e}}|\vec{\mathrm{s}})\rangle)$.
	If $\vec{\mathrm{c}}=^t\vec{\mathrm{xA}}+\vec{\mathrm{y}}$, we have $\langle \vec{\mathrm{c}},\vec{\mathrm{s}} \rangle=\langle \vec{\mathrm{x}},\vec{\mathrm{b}} \rangle +\langle \vec{\mathrm{x}},\vec{\mathrm{e}} \rangle+\langle \vec{\mathrm{y}},\vec{\mathrm{s}} \rangle$ so that the sample is distributed according to $\Dis_1$.
	Therefore, we have a distinguisher for $\mathsf{Extended\mhyphen LWE}$ of advantage $\frac{\epsilon_0}{k}$.

	Clearly, there is a distinguisher of advantage $\epsilon_1$ against $\mathsf{Small\mhyphen Decision\mhyphen LWE}$. 

	Let $(\vec{\mathrm{A}},\vec{\mathrm{c}},\vec{\mathrm{u}},\vec{\mathrm{r}})$ be an instance of the $\mathsf{First\mhyphen is\mhyphen errorless\mhyphen LWE}$ problem.
	The hybrid lemma with distributions $\Dis_2$ and $\Dis_3$, and sample $(\vec{\mathrm{c}},\vec{\mathrm{r}})$ shows that there exist an efficient distinguisher for $\mathsf{First\mhyphen is\mhyphen errorless\mhyphen LWE}$ with advantage $\epsilon_3/k$.

	Since $\epsilon_0+\epsilon_1+\epsilon_2 \geq \epsilon$, using the previous lemmata, and the fact that $\DecisionLWE$ is harder in dimension $m$ than $m-1$, the theorem follows.

	For the search version, the indistinguishability of $\Dis_0$ and $\Dis_1$ is sufficient.
\end{proof}

D{\"o}ttling's reduction has a uniform secret, so that a problem solvable in time $2^{\bigo(\sqrt{n})}$ is transformed into a problem where the best algorithm takes time $2^{(1+o(1))n/\log n}$.
We do not have such dramatic loss here.

\subsubsection{Applications}

\begin{lemma}
	Let $\Dis=D_{\Z,\sigma}$ for $\sigma \geq 1$.

	Then, the advantage of a distinguisher for $\DecisionLWE$ of dimension $m$ with $m+n$ samples of noise distribution $\Dis$ is at most $\sqrt{q^n/\sigma^{n+m}}$.
	Furthermore, the bias of $\langle (\vec{\mathrm{s}}|\vec{\mathrm{e}}),(\vec{\mathrm{x}}|\vec{\mathrm{y}}) \rangle$, for fixed $\vec{\mathrm{s}}$ and $\vec{\mathrm{e}}$, is at least $\exp(-\pi(||\vec{\mathrm{s}}||^2+||\vec{\mathrm{e}}||^2)\sigma^2/q^2)$.
\end{lemma}
\begin{proof}
	We have $\Dis^{m+n}(\vec{\mathrm{a}})\leq \Dis(0)^{m+n}=1/\rho_{\sigma}(\Z)^{m+n}$ and $\rho_{\sigma}(\Z)=\sigma \rho_{1/\sigma}(\Z) \geq \sigma$ using a Poisson summation.
	The first property is then a direct application of the leftover hash lemma, since $q$ is prime.

	The bias of $\lambda \Dis$ can be computed using a Poisson summation as :
	\[ \sum_{a \in \Z} \rho_{\sigma}(a)\cos(2\pi \lambda a/q)=\rho_{1/\sigma}(\Z+\lambda/q)\geq \exp(-\pi\lambda^2\sigma^2/q^2). \]
	Therefore, the second property follows from the independency of the coordinates of $\vec{\mathrm{x}}$ and $\vec{\mathrm{y}}$.
\end{proof}

\begin{corollary}
	Let $q$, $n$ and $m$ such that $(m-3)\log q/(n+m)-\log k=\omega(1)$ and $m\log q/(n+m)=o(n/log n)$.
	Then, we can solve $\mathsf{Small\mhyphen Decision\mhyphen LWE}$ in time \[ 2^{(n/2+o(n))/\ln(1+((m-3)\log q/(n+m)-\log k)/\log B)} \] with negligible probability of failure.
\end{corollary}
\begin{proof}
	We use the previous lemma with $\sigma=2q^{(n+2)/(n+m-1)}$, so that we have $\beta=\Omega(q^{(m-3)/(n+m)}/k)$.
	The algorithm from \autoref{thm:lweL2} needs $2^{o(n)}$ samples, so that the advantage of the potential distinguisher for $\DecisionLWE$ is $2^{-n/4+o(n)}/q$ for $\xi=2^{-n/4}$ ; while the previous lemma proves it is less than $2^{-n/2}/q$.
\end{proof}

The \textsf{NTRU} cryptosystem~\cite{hoffstein1998ntru} is based on the hardness of finding two polynomials $f$ and $g$ 
whose coefficients are bounded by $1$ given $h=f/g \mod (X^n-1,q)$. Since $hg=0$ with an error bounded by $1$, 
we can apply previous algorithms in this section to \textit{heuristically} recover $f$ and $g$ in time $2^{(n/2+o(1))/\ln \ln q}$. This is the first subexponential
time algorithm for this problem since it was introduced back in 1998. 
\begin{corollary}
	Assume we have a $\SearchLWE$ problem with at least $n\log q+r$ samples and Gaussian noise with $\alpha=n^{-c}$ and $q=n^d$.
	Then, we can solve it in time $2^{n/(2\ln(d/(d-c))-o(1))}$ for any failure probability in $2^{-n^{o(1)}}+q^{-r}$.
\end{corollary}
\begin{proof}
	First, apply a secret-error switching (\autoref{th:petitsecret}) and assume we lose at most $r$ samples.
	Apply the previous corollary with $B=n^{d-c+o(1)}$ which is a correct bound for the secret, except with probability $2^{-n^{o(1)}}$.
	\autoref{lemma:ban93} shows that $k^2\leq \log q\sigma^2$, except with probability $2^{-\mathrm{\Omega}(n)}$, so that $\beta=n^{c+o(1)}$.
	We can then use $\sigma=\Theta(1)$ and apply~\autoref{thm:lwe}.
\end{proof}
Note that this corollary can in fact be applied to a very large class of distributions, and in particular to the learning with rounding problem, while the distortion parameter is too large for a direct application of~\autoref{thm:lwe}.

Also, if the reduction gives a fast (subexponential) algorithm, one may use $\sigma=2\sqrt{n}$ and assume that there is no quantum algorithm solving the corresponding lattice problem in dimension $m$.

Even more heuristically, one can choose $\sigma$ to be the lowest such that if the reduction does not work, we have an algorithm faster than the best {\it known} algorithm for the same problem.

\bibliographystyle{abbrv}
\bibliography{ref}

\appendix

\section{$\mathsf{LWE}$ in Small Dimension}
If the dimension $n$ is small ($n=\bigo(\log(\beta))$), and in particular if $n$ is one~\footnote{This is an interesting 
case in cryptography for analyzing the security of DSA\cite{de2013using} or proving the hardness of recovering the most 
significant bits of a Diffie-Hellman.}, we can use the same kind of algorithm. Such algorithms have been described 
in~\cite{de2013using,DBLP:conf/asiacrypt/AranhaFGKTZ14} and an unpublished paper of Bleichenbacher.

The principle is to reduce at each step some coordinates and partially reduce another one.
If $q$ is too large, then the reduction of the last coordinate has to be changed: we want to have this coordinate {\it uniformly} distributed in the ball of dimension 1 and radius $R$ (i.e. $[-R;R]$).
To this end, for each sample, if the ball of radius $R'$ centered at the sample is included in the ball of radius $R$ and if it contains at least another sample, 
we will choose uniformly a sample in the ball~\footnote{We can do it in logarithmic time using a binary balanced tree which contains in each node the size of the subtree.} and then remove it and add the difference in the output list. In the other case, we add the element to our structure.
The last coordinate of the outputted samples is clearly uniform over a ball of radius $R'$.

Finally, we apply a Fourier transform to the samples : for various potential secrets $x$, if we denote by $(a,b)$ the samples where $a$ is the last coordinate of the vectorial part and $b$ the scalar part, we compute the experimental bias $\Re(\E[\exp(2i\pi(ax-b)/q)])$ and output the $s$ which maximizes them.

We can analyze the Fourier stage of the algorithm with $a$ uniformly distributed in a ball of radius $R$, an odd integer.
The bias for $x$ is $\E[\exp(2i\pi(e+(s-x)a)/q)]$ where $e=as-b$ is the error, which is equal to the bias of the error times $\E[\exp(2i\pi(s-x)a/q)]$.
We take $x$ regularly distributed so that the interval between two $x$ is smaller than $q/R/2$, and we can compute all bias in one Fourier transform.

Let us consider an $x$ such that $|s-x|\leq q/R/4$. The introduced bias is the bias of the error times 
$\frac{\sin \pi(s-x)R/q}{R\sin \pi(s-x)/q}$, which is lower bounded by a universal constant. However, for $|s-x|\geq q/R$, 
the introduced bias is upper bounded by a constant, strictly smaller than the previous one. 
Therefore, the Fourier stage determines a $x$ such that $|s-x|<q/R$.

We will in a first time determine the secret up to an approximation factor $q/R$ using a fast Fourier transform on 
$\approx 2R$ points. We are therefore reduced to a smaller instance where the secret is bounded by $q/R$ and we 
can choose $R'=R^2$ for instance. Using one more time a Fourier Transform on $2R'/R=2R$ points, we can recover 
the secret up to an approximation factor of $q/R'=q/R^2$. We can continue so on, and we get the final coordinate
of the secret and repeating this process $n-1$ times, we have the whole secret. 

We define $\beta=\sqrt{\log q}/\alpha$ and express all variables as functions of $q$.

\begin{theorem}
	We can solve $\SearchLWE$ with $n=1$, $\beta=\omega(1)$, $\beta=q^{o(1/\log q)}$ and distortion parameter $\epsilon\leq \beta^{-4}$ in time 
	\[ q^{(1/2+o(1))/\log \beta} \]
	with failure probability $q^{-\Omega(1)}$.
\end{theorem}
\begin{proof}
	We analyze here the first iteration of the algorithm.

	We set $k=\lfloor \log(\beta^2/\log \beta/3) \rfloor$ the number of iterations, $m=\Theta(\log q2^kq^x)$ the number of input samples and $R_i=q/R^i$ such that the vectorial part of the input of the $i$th reduction is sampled uniformly in a ball of the largest odd radius smaller than $R_i$.
	If we have $N$ samples, the Hoeffding inequality shows that one reduction step outputs at least $N(1/2-1/R-\log q/R^2)-R$ samples, except with probability $q^{-\Omega(1)}$.
	Therefore, we can use $R=q^{x}$.
	For simplicity, we then set $R_k=3$ so that $R \leq q^{1/a}$.
	The bias of the error is at least $\exp(-2^k\alpha^2)(1-\epsilon 3^k) \geq \exp(-\log q/\log \beta/3)/2$ so that taking $x=1/k$ works.

	The overall complexity is $\bigo(\log^3 q\beta^2 q^{1/k})$ which is $q^{(1/2+o(1))/\log \beta}$.
\end{proof}

\section{Solving $\mathsf{LWE}$ with Lattice Reduction}
\label{app:LLLbinary}
Lattice reduction consists in finding a basis of a lattice with short vectors.
The best algorithm\cite{gama2008finding} uses a polynomial number of times an oracle which returns the shortest non-zero vector of a lattice of dimension $d$, and finds a vector whose norm is inferior to $(\gamma_d+o(1))^{(n-1)/(d-1)} V^{1/n}$ for a lattice of volume $V$ and dimension $n$, with $\gamma_d\leq 2B_d^{1/d}$ the square of the Hermite constant.

Heuristically, we use the BKZ algorithm which is stopped after a polynomial number of calls to the same oracle.
The $i$-th vector of the orthonormalized basis has a length proportional to roughly $\gamma^i$ and the Gaussian heuristic gives $B_d=\gamma^{d(d+1)/2}$ so that $\gamma \approx B_d^{2/d^2}$ and the first vector has a norm $\approx \gamma^{-n/2} V^{1/n} \approx B_d^{n/d^2} V^{1/n} \approx d^{n/2/d} V^{1/n}$.

The oracle takes time $2^{(c+o(1))d}$ with $c=2$ for a deterministic algorithm\cite{micciancio2013deterministic}, $c=1$ for a randomized algorithm\cite{aggarwal2014solving}, $c=0.2972$ heuristically\cite{laarhoven2014sieving} and $c=0.286$ heuristically and quantumly\cite{laarhoven2014solving}.

\subsection{Primal algorithm}
Given $\vec{\mathrm{A}}$ and $\vec{\mathrm{b}}=\vec{\mathrm{As}}+\vec{\mathrm{e}}+q\vec{\mathrm{y}}$, with the coordinates of $\vec{\mathrm{s}}$ and $\vec{\mathrm{e}}$ being Gaussians of standard deviation $q\alpha=o(q)$, we use the lattice generated by
\[ \vec{\mathrm{B}}=\begin{pmatrix} q\vec{\mathrm{I}} & \vec{\mathrm{A}} \\ \vec{\mathrm{0}} & \vec{\mathrm{I}} \end{pmatrix} \]
Now, \[ \vec{\mathrm{B}}\begin{pmatrix} \vec{\mathrm{y}} \\ \vec{\mathrm{s}}\end{pmatrix}=\begin{pmatrix} \vec{\mathrm{b}} \\ \vec{\mathrm{0}}\end{pmatrix}+\begin{pmatrix} -\vec{\mathrm{e}} \\ \vec{\mathrm{s}}\end{pmatrix} .\]
We then heuristically assume that the $i$-th vector of the orthonormalized basis has a length at least $d^{-(n+m)/2/d}V^{1/(n+m)}$.
In the orthogonalized basis, the error vector has also independent gaussian coordinates of standard deviation $\alpha q$.
Therefore, Babai algorithm will work with high probability if $d^{-(n+m)/2/d}V^{1/(n+m)} \geq \alpha q\sqrt{\log n}$.

Thus, we need to maximize $d^{-(n+m)^2/2/d}q^{m/(n+m)}$.
So we use $m+n = \sqrt{2dn\log q/\log d}$ and the maximum is \[ 2^{-\sqrt{n\log q\log d/2/d}}2^{\log q(1-\sqrt{n\log d/2/d/\log q})}=q2^{-\sqrt{2n\log q\log d/d}} .\]
If $\alpha=n^{-a}$ and $q=n^b$, we need $a \log n=\sqrt{2bn\log n \log d/d}$ and we deduce $d/n=2b/a^2+o(1)$, so that the running time is $2^{(2cb/a^2+o(1))n}$.

When we have a binary secret and error, since only the size of the vectors counts, it is as if $q\alpha=\bigo(1)$ and the complexity is therefore $2^{(2c/b+o(1))n}$.
When the secret is binary and the error is a Gaussian, the problem is clearly even more difficult\footnote{In this case, a standard technique is to scale the coordinates of the lattice corresponding to $\vec{\mathrm{e}}$, in order to have each coordinate of the distance to the lattice of the same average size.}.

In dimension $1$, this gives $\log 1/\alpha=\sqrt{2\log q\log d/d}$ so that $d=(2+o(1)) \log q\log(\log q/\log^2 \alpha)/\log^2 \alpha$, which gives a polynomial time algorithm for $\log^2 \alpha=\Omega(\log q \log \log \log q/\log \log q)$ and a subexponential one for $\log^2 \alpha=\omega(\log \log q)$.

We can use the same algorithm for subset-sum with density $d$ with the lattice defined in~\autoref{sec:subsetsum} : applying the same heuristics gives a complexity of $2^{(c+o(1))nd\log(nd)}n^{\bigo(1)}$.

\subsection{Dual algorithm}
The attack consists in using lattice reduction to find sums of samples which are equal to zero on most coordinates, and then use {\sc FindSecret} to find the secret on the remaining coordinates. It has been described for instance in~\cite{de2013using}.

More precisely, we select $n$ truncated samples which form an invertible matrix $\vec{\mathrm{A}}$ of size $n$; and $m$ others which form the matrix $\vec{\mathrm{B}}$ with $m$ columns.
We then search for a short vector $\vec{\mathrm{v}}\neq \vec{\mathrm{0}}$ of the lattice generated by $\begin{pmatrix} q\vec{\mathrm{I}} && \vec{\mathrm{A}}^{-1}\vec{\mathrm{B}} \\ \vec{\mathrm{0}} && \vec{\mathrm{I}}\end{pmatrix}$.

We have $\vec{\mathrm{v}}=\begin{pmatrix} \vec{\mathrm{x}} \\ \vec{\mathrm{y}} \end{pmatrix}$ and $\vec{\mathrm{A}}(-\vec{\mathrm{x}})+\vec{\mathrm{B}}\vec{\mathrm{y}}=\vec{\mathrm{A}}(-\vec{\mathrm{x}}+\vec{\mathrm{A}}^{-1}\vec{\mathrm{B}}\vec{\mathrm{y}})=\vec{\mathrm{A}}\vec{0}=\vec{0} \mod q$ so we deduce from $\vec{\mathrm{v}}$ a small sum of samples which is equal to zero on $n$ coordinates.
We can then deduce a sample of low dimension with bias $\me^{-||\vec{\mathrm{v}}||^2\alpha^2}$.
As $||\vec{\mathrm{v}}||$ is inferior to $\approx q^{n/(n+m)}2^{(n+m)\log d/d/2}$, so we select $n+m=\sqrt{2nd\log q/\log d}$ and $||\vec{\mathrm{v}}||^2 \approx 2^{\sqrt{8n\log q \log d/d}}$.

Finally, using $\alpha=n^{-a}$, $q=n^b$ and $d=\Theta(n)$, we search $d$ so that $2^{\sqrt{8n\log q \log d/d}}=\bigo(d/\alpha^2)$ so that $d/n = 8b/(1+2a)^2+o(1)$.

When the errors are binary, we can apply the same formula with $\alpha=\Theta(1/q)$ and the complexity is therefore $2^{(2cb/(1/2+b)^2+o(1))n}$.

Observe that this algorithm is always asymptotically faster than the previous one, but requires an exponential number of samples.
This disadvantage can be mitigated by finding many short vectors in the same lattice, which hopefully will not affect the behaviour of {\sc FindSecret}.
Another solution is to use a sample preserving search-to-decision reduction \cite{micciancio2011pseudorandom} so that the product of the time and the square of the success probability is the same.

In dimension $1$, this gives $\log (d/\alpha^2)=\sqrt{8\log q \log d/d}$ and we deduce \[ d=(8+o(1)) \log q \log\big(\log q/\log^2(d/\alpha^2)\big)/\log^2 (d/\alpha^2) \]
and this is a polynomial time algorithm for $\log^2 \alpha=\Omega(\log q \log \log \log q/\log \log q)$, but a subexponential one for any $\log(q)/\alpha^2=\omega(2^{\sqrt{\log \log q}})$.

\subsection{Graphical Comparison between lattice algorithms and BKW}

\begin{center}
\noindent\makebox[\textwidth]{%
\includegraphics[scale=.28]{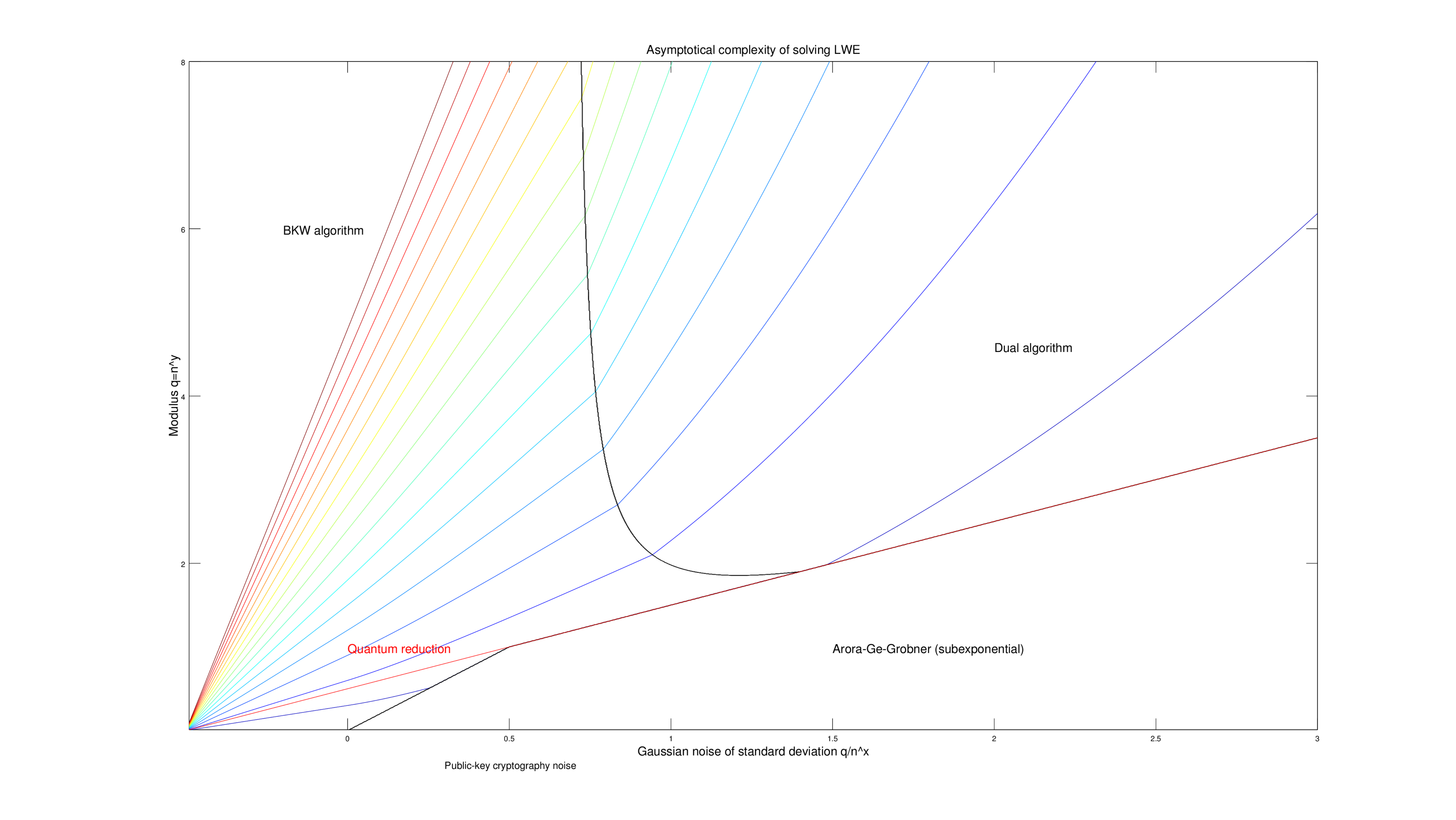}}
\end{center}

The horizontal axis represents the standard deviation of the noise, with $x$ such that $\alpha = n^{x}$.
The vertical axis represents the modulus, with $y$ such that $q = n^{y}$.

In the upper right quadrant delimited by the black line, parameters are such that the dual attack with heuristic classical SVP is more efficient than BKW.
Below the black lines, Gr\"obner basis algorithms are sub-exponential.
Above the red line, there is a quantum reduction with hard lattice problems~\cite{regev2009lattices}.

The rainbow-colored lines are contour lines for the complexity of the best algorithm among BKW and the dual attack. Each curve follows the choice of parameters such that the overall complexity is $2^{(k+o(1))n}$, where $k$ varies from 0 to 5 by increments of 0.3.

\begin{center}
\noindent\makebox[\textwidth]{%
\includegraphics[scale=.28]{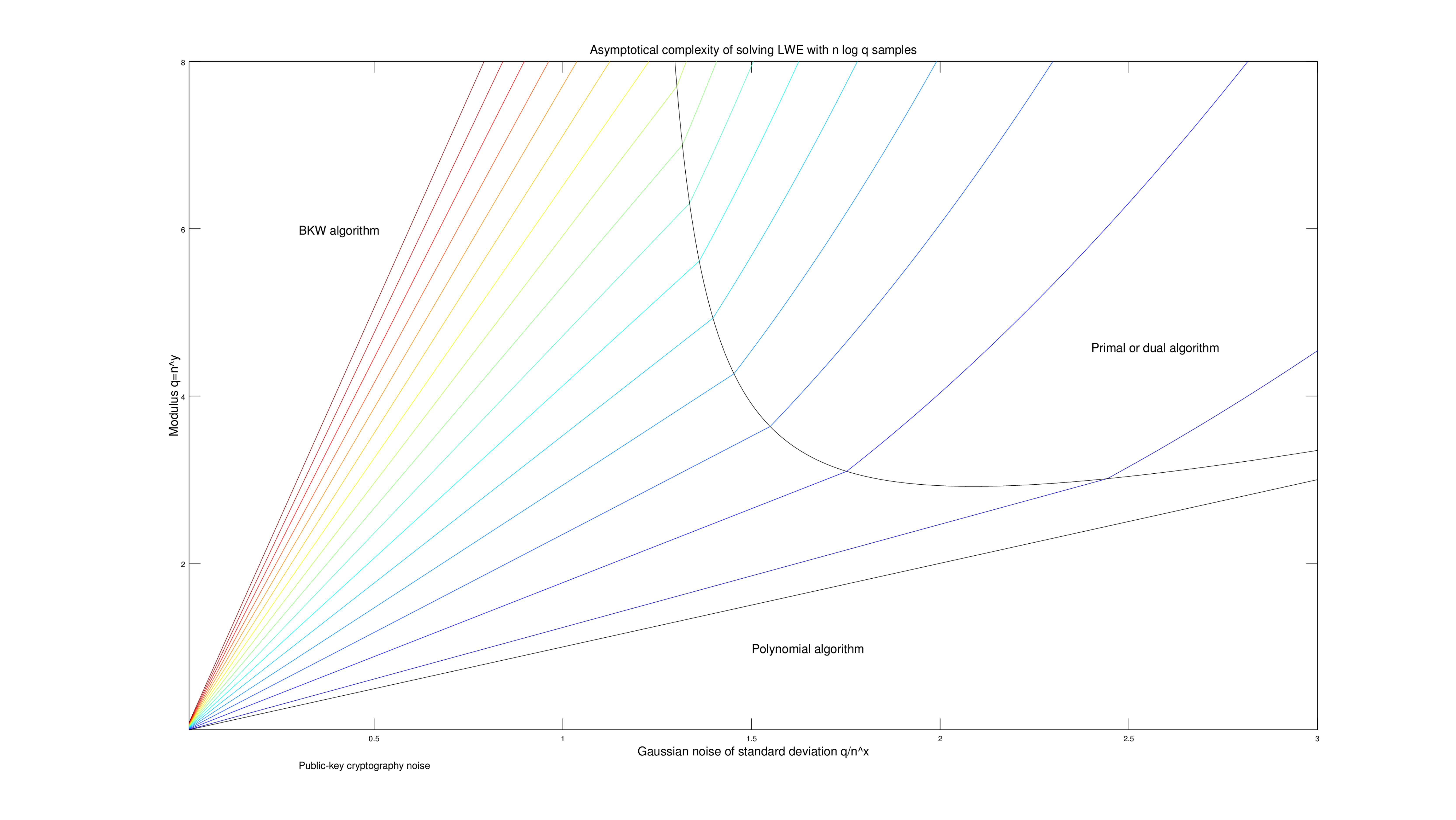}}
\end{center}

\section{$\mathsf{LPN}$}

\begin{definition}
	A $\mathsf{LPN}$ distribution is defined as a $\mathsf{LWE}$ distribution with $q=2$. The error follows a Bernoulli distribution of parameter $p$.
\end{definition}

\begin{theorem}
	If we have $p=1/2-2^{-o(n)}$ and $p\geq 1/n$, then we solve $\mathsf{Decision\mhyphen LPN}$ in time \[ 2^{(n+o(n))/\log(n/\log(1-2p))} \] with negligible probability of failure.
\end{theorem}
\begin{proof}
	We use $d_i=\min(i\lfloor nx\rfloor,n)$, $m=n2^{k+nx}$ and select $k=\lfloor \log(\beta^2/\log \beta/3) \rfloor$ where $\beta^2=n/2/\log (1-2p)=\omega(1)$.
	The algorithm works if $d_k=n$, so we chose $x=1/k+1/n=(1+o(1))/\log(\beta^2)$.
	Finally, the bias of the samples given to {\sc Distinguish} is $2^k\alpha^2\leq \beta^2/\log \beta/3 n/2/\beta^2 =n/6/\log \beta \leq nx/3$ so that {\sc Distinguish} works with negligible probability of failure.
\end{proof}
In particular, for $p=1/2-2^{-n^{o(1)}}$, we have a complexity of $2^{(n+o(n))/\log n}$.

\begin{theorem}
	If we have $p\leq 1/4$, then we solve $\mathsf{Search\mhyphen LPN}$ in time $\bigo(n^4(1-p)^{-n})$ with negligible probability of failure.
\end{theorem}
\begin{proof}
	We repeat $n(1-p)^{-n}$ times the following algorithm : apply the secret-error switching algorithm (\autoref{th:petitsecret}) over $34n$ new samples, and run {\sc Distinguish} over the scalar part of the $33n$ samples.
	If the distinguisher outputs that it is a $\mathsf{LPN}$ distribution, then the secret is the zero vector with failure probability $2^{-33/32n}$, and we can recover the original secret from this data.
	If the new secret is not zero, then the distinguisher outputs that the distribution is uniform, except with probability $2^{-33/32n}$.
	The probability that no new secret is equal to zero is 
	\[ \big(1-(1-p)^n\big)^{n(1-p)^{-n}} \leq  \exp(-(1-p)^nn(1-p)^{-n})=2^{-\Omega(n)}. \]
	Using the union bound, the overall probability of failure is at most $n^42^n2^{-33/32n}=2^{-\Omega(n)}$.
\end{proof}
In particular, for $p=\frac{1}{\sqrt{n}}$, the complexity is $\bigo(n^4\exp(\sqrt{n}))$.
In fact, for $p=(1+o(1))/\ln(n) \leq 1/4$, we have a complexity which is $\bigo(n^4\exp(n(p+p^2)))=2^{(1+o(1))n/\log n}$, so that the threshold for $p$ between the two algorithms is in $(1+o(1))/\ln n$.

\end{document}

%% file: intro.tex
The Learning With Errors ($\mathsf{LWE}$) problem has been an important problem in cryptography since its introduction by Regev in~\cite{regev2009lattices}. Many cryptosystems have been proven secure assuming the hardness of this
problem, including Fully Homomorphic Encryption schemes~\cite{gentry2013homomorphic,DBLP:journals/siamcomp/BrakerskiV14}.
The decision version of the problem can be described as follows: given $m$ samples of the form 
$(\vec{\mathrm{a}},b)\in (\mathbb{Z}_q)^n\times \mathbb{Z}_q$, where $\vec{\mathrm{a}}$ are uniformy distributed in $(\mathbb{Z}_q)^n$, 
distinguish whether $b$ is uniformly chosen in $\mathbb{Z}_q$ or is equal to $\langle\vec{\mathrm{a}},\vec{\mathrm{s}}\rangle+e$ for a fixed secret 
$\vec{\mathrm{s}} \in (\mathbb{Z}_q)^n$ and $e$ a noise value in $\mathbb{Z}_q$ chosen according to some probability 
distribution. Typically, the noise is sampled from some distribution concentrated on small numbers, such as a discrete 
Gaussian distribution with standard deviation $\alpha q$ for $\alpha=o(1)$. 
In the search version of the problem, the goal is to recover $\vec{\mathrm{s}}$ given the promise that the sample instances 
come from the latter distribution. Initially, Regev showed that if $\alpha q \geq 2\sqrt{n}$, 
solving $\mathsf{LWE}$ on average is at least as hard as approximating lattice problems in the worst case 
to within $\tilde{\mathcal{O}}(n/\alpha)$ factors with a quantum algorithm. Peikert shows a classical 
reduction when the modulus is large $q\geq 2^n$ in~\cite{peikert2009public}. 
Finally, in~\cite{brakerski2013classical}, Brakerski \textit{et al.} prove that solving 
$\mathsf{LWE}$ instances with polynomial-size modulus in polynomial time implies an efficient solution to $\mathsf{GapSVP}$.

There are basically three approaches to solving $\mathsf{LWE}$: the first relies on lattice reduction techniques such 
as the $\mathsf{LLL}$~\cite{LLL} algorithm and further improvements~\cite{DBLP:conf/asiacrypt/ChenN11} as 
exposed in~\cite{DBLP:conf/ctrsa/LindnerP11,DBLP:conf/ctrsa/LiuN13}; the second uses combinatorial 
techniques~\cite{BKW,Wagner}; and the third uses algebraic techniques~\cite{DBLP:conf/icalp/AroraG11}.
According to Regev in~\cite{DBLP:conf/coco/2010}, the best known algorithm to solve $\mathsf{LWE}$ is 
the algorithm by Blum, Kalai and Wasserman in~\cite{BKW}, originally proposed to solve the Learning Parities with Noise ($\mathsf{LPN}$) problem, 
which can be viewed as a special case of $\mathsf{LWE}$ where $q=2$.  
The time and memory requirements of this algorithm are both exponential for $\mathsf{LWE}$ and subexponential for 
$\mathsf{LPN}$ in $2^{\bigo(n/\log n)}$. During the first stage of the algorithm, the dimension of 
$\vec{\mathrm{a}}$ 
is reduced, at the cost of a (controlled) decrease of the bias of $b$. During the second stage, the algorithm distinguishes between 
$\mathsf{LWE}$ and uniform by evaluating the bias. 

Since the introduction of $\mathsf{LWE}$, some variants of the problem have been proposed in order to build more efficient 
cryptosystems. Some of the most interesting variants are $\mathsf{Ring\mhyphen LWE}$ by Lyubashevsky, Peikert and Regev 
in~\cite{lyubashevsky2013ideal}, which aims to reduce the space of the public key using cyclic samples; and the 
cryptosystem by D\"ottling and M\"uller-Quade~\cite{DBLP:conf/eurocrypt/DottlingM13}, which uses short 
secret and error. In 2013, Micciancio and Peikert~\cite{micciancio2013hardness} as well as 
Brakerski \textit{et al.}~\cite{brakerski2013classical} proposed a binary version of the $\mathsf{LWE}$ problem and 
obtained a hardness result. 

{\bf Related Work.}
Albrecht \textit{et al.} have presented an analysis of the BKW algorithm as applied to $\mathsf{LWE}$ in~\cite{albrecht2013complexity,albrecht2014lazy}. 
It has been recently revisited by Duc \textit{et al.}, who use a multi-dimensional FFT in the second stage of 
the algorithm~\cite{duc2015better}. However, the main bottleneck is the first BKW step and since the proposed algorithms do not improve this stage, the overall asymptotic complexity is unchanged.

In the case of the $\mathsf{BinaryLWE}$ variant, where the error and secret are binary (or sufficiently small), 
Micciancio and Peikert show that solving this problem using $m=n(1+\mathrm{\Omega}(1/\log(n)))$ samples
is at least as hard as approximating lattice problems in the worst case in dimension $\Theta(n/\log(n))$
with approximation factor $\tilde{\mathcal{O}}(\sqrt{n} q)$. We show in \autoref{app:LLLbinary} that 
existing lattice reduction techniques require exponential time. 
Arora and Ge describe a $2^{\tilde{\mathcal{O}}(\alpha q)^2}$-time algorithm when $q>n$ to solve the $\mathsf{LWE}$ problem~\cite{DBLP:conf/icalp/AroraG11}. This leads to a subexponential time algorithm when the error magnitude 
$\alpha q$ 
is less than $\sqrt{n}$. The idea is to transform this system into a noise-free polynomial system and then use 
root finding algorithms for multivariate polynomials to solve it, using either relinearization in~\cite{DBLP:conf/icalp/AroraG11} 
or Gr\"obner basis in~\cite{albrechtalgebraic}. In this last work, Albrecht \textit{et al.} present an algorithm whose 
time complexity is $2^{\frac{(\omega+o(1)) n \log \log \log n}{8\log \log n}}$ when the number of samples $m=(1+o(1))n \log \log n$ is 
super-linear, where $\omega < 2.3728$ is the linear algebra constant, under some assumption 
on the regularity of the polynomial system of equations; and when $m=\bigo(n)$, the complexity becomes exponential.  

{\bf Contribution.}
Our first contribution is to present in a unified framework the BKW algorithm and all its previous improvements in the binary case~\cite{parity,kirchner2011improved,bernstein2013never,guo2014solving} and in the general case~\cite{albrecht2014lazy}.
We introduce a new quantization step, which generalizes modulus 
switching~\cite{albrecht2014lazy}.
This yields a significant decrease in the constant of the exponential of the complexity for $\mathsf{LWE}$.
Moreover our proof does not require Gaussian noise, and does not rely on unproven independence assumptions.
Our algorithm is also able to tackle problems with larger noise.

We then introduce generalizations of the $\mathsf{BDD}$, $\mathsf{GapSVP}$ and $\mathsf{UniqueSVP}$ problems, and prove a reduction from these variants to $\mathsf{LWE}$.
When particular parameters are set, these variants impose that the lattice point of interest (the point 
of the lattice that the problem essentially asks to locate: for instance, in the case of $\mathsf{BDD}$, the point of the 
lattice closest to the target point) lie in the fundamental parallelepiped; or more generally, we ask that the coordinates 
of this point relative to the basis defined by the input matrix $\vec{\mathrm{A}}$ has small infinity norm, bounded by 
some value $B$. For small $B$, our main algorithm yields a subexponential-time algorithm for these variants of $\mathsf{BDD}$, $\mathsf{GapSVP}$ and $\mathsf{UniqueSVP}$.

Through a reduction to our variant of $\mathsf{BDD}$, we are then able to solve the subset-sum problem in subexponential time when the density is $o(1)$, and in time $2^{(\ln 2/2+o(1))n/\log \log n}$ if the density is $\bigo(1/\log n)$. This is of independent interest, as existing techniques for density $o(1)$, based on lattice reduction, require exponential time.
As a consequence, the cryptosystems of Lyubashevsky, Palacio and Segev at TCC 2010~\cite{lyubashevsky2010public} can be solved in subexponential time. 

As another application of our main algorithm, we show that $\mathsf{BinaryLWE}$ with reasonable noise can be solved in time $2^{(\ln 2/2+o(1))n/\log \log n}$ instead of $2^{\Omega(n)}$; and the same complexity holds for secret of size up to $2^{\log^{o(1)} n}$. As a consequence, we can heuristically recover the secret polynomials $\vec{\mathrm{f}},\vec{\mathrm{g}}$ of the $\mathsf{NTRU}$ problem in subexponential time $2^{(\ln 2/2+o(1))n/\log \log n}$ (without contradicting its security assumption). The heuristic assumption comes from the fact that $\mathsf{NTRU}$ samples are not random, since they are rotations of each other: the heuristic assumption is that this does not significantly hinder $\mathsf{BKW}$-type algorithms. Note that there is a large value hidden in the $o(1)$ term, so that our algorithm does not yield practical attacks for recommended $\mathsf{NTRU}$ parameters.